\documentclass[11pt]{article}
\usepackage{amsmath, amsfonts, amssymb, amsthm, subcaption, booktabs}
\usepackage[margin=1in]{geometry}
\usepackage{graphicx, color}
\usepackage{url}
\usepackage{mathabx}
\usepackage{algorithm}
\usepackage{algpseudocode}

\numberwithin{equation}{section}
\newtheorem{remark}{Remark}
\newtheorem{lemma}{Lemma}

\newcommand{\gb}{{\mathbf{g}}}
\newcommand{\bb}{{\mathbf{b}}}

\newcommand{\psib}{\boldsymbol\psi}
\newcommand{\phib}{\boldsymbol\phi}
\newcommand{\sigmab}{\boldsymbol\sigma}

\providecommand{\keywords}[1]
{
  \small	
  \textbf{\textit{Keywords---}} #1
}

\usepackage{authblk}

\author[1]{N.~Anders Petersson~\thanks{petersson1@llnl.gov}}
\author[2]{Chase Hodges-Heilmann}
\author[3]{Stefanie G\"unther}

\affil[1,3]{Center for Applied Scientific Computing, Lawrence Livermore National Laboratory}
\affil[2]{Department of Mathematics, University of New Mexico}

\title{Dynamical Simulations of Schr\"odinger's Equation via Rank-Adaptive Tensor Decompositions}

\date{May 1, 2026}

\begin{document}

\maketitle
\bibliographystyle{unsrt}

\begin{abstract}
We study low-rank tensor methods for the numerical solution of Schr\"odinger's equation with time-independent and explicitly time-dependent Hamiltonians, motivated by large-scale simulations of many-body quantum systems and quantum computing devices subject to time-dependent control pulses. We outline the recent application of the "basis update and Galerkin" (BUG) method for tensor trains, and describe the established TDVP and TDVP-2 algorithms based on the time-dependent variational principle. For comparison, we also consider the BUG method in the Tucker format. All these approaches enable memory efficient representations of partially entangled quantum states and thereby mitigate the exponential cost of conventional state-vector formulations. The rank-adaptivity relies on the truncated singular value decomposition, in which the rank of a matrix is reduced by setting its smallest singular values to zero, based on a threshold parameter that controls the truncation error. 
Numerical experiments on representative time-independent and time-dependent Hamiltonian models quantify the tradeoff between accuracy and compression across methods, with particular attention to the interplay between the time-step and the truncation threshold, and how the computational effort scales with the number of sub-systems in the quantum system.
\end{abstract}


\section{Introduction}\label{sec_intro}

Classical simulations of many-body quantum systems and quantum computing devices generally become intractable as the number of coupled sub-systems (qubits) increase. This is due to the exponential growth of the quantum state vector and the associated increase in computational effort, as the numbers of sub-systems increase. For example, in the case of a quantum system consisting of $N$ two-level sub-systems, the state vector has $2^N$ elements, and the Hamiltonian is an operator of size $2^N \times 2^N$. Storing and manipulating these objects becomes intractable, even for relatively modest numbers of sub-systems.
An alternative perspective is to represent the quantum state as a high-dimensional tensor,
$|\psi\rangle \in \mathbb{C}^{d_1 \times d_2 \times \cdots \times d_N}$, where each mode corresponds to a sub-system. While the full tensor representation contains the same number of entries as the state vector, it enables the use of tensor decomposition techniques that allow for a compressed representation of the state.

Compact parameterizations of the quantum state based on a network of low-rank tensors, e.g., matrix-product states (MPS) (also known as tensor-trains), have been used in several quantum simulation applications. For example, MPS methods were used in gate-based circuit simulations to evaluate observables in a kicked Ising model on 127 qubits~\cite{BegGraCha-2024}. Another application of MPS is the recent classical computation of the ground state energy, to chemical precision, in the 76-orbital/152-qubit FeMo-cofactor model, based on the DMRG method~\cite{zhai-2026}. In terms of Hamiltonian simulations, the MPS representation has also emerged as a powerful tool for simulating the dynamics of quantum many-body systems, providing an efficient way to represent entangled quantum states in large Hilbert space dimensions. Methods such as the time-evolving block decimation method (TEBD)~\cite{Daley-2004} and algorithms based on the time-dependent variational principle~\cite{Oseledets-2011, HaeCirOsb-2011, HaeLubOse-2016, PAECKEL2019167998} have been developed to study the dynamical behavior of such systems, where the state is represented as an MPS. Existing approaches have focused on simulating the quantum dynamics by solving the time-dependent Schrödinger equation for a given time-independent Hamiltonian, e.g., a linear vibronic model including 252 normal modes~\cite{SchTurMus-2019}.

In this paper we broaden the application of tensor decomposition methods to dynamical simulations of Schr\"odinger's equation when the Hamiltonian is time-dependent, e.g., to study quantum computing devices subject to time-dependent control pulses. Accurate dynamical simulation is a key component for optimal control in quantum computing and quantum error-correction, e.g., to perform state-to-state transformations, or to implement unitary gates. By concatenating the control pulses in time for all individual gates in a quantum algorithm, dynamical simulations will also allow the overall fidelity of an algorithm to be predicted when applied to a quantum device. 

In the following we focus on time-integration of two tensor decompositions: Tucker tensors and tensor trains. We describe the nascent application of the "basis update and Galerkin" (BUG) method for tensor trains~\cite{Sulz-2026}, as well as the classical TDVP and TDVP-2 algorithms~\cite{PAECKEL2019167998}. For Tucker tensors we use the BUG method developed by Ceruti et al.~\cite{Ceruti2022}. The rank-adaptivity in the TDVP-2 and BUG algorithms relies on the truncated singular value decomposition (SVD), in which the rank of a matrix is reduced by setting its smallest singular values to zero. The rank is determined such that the error due to the truncation is smaller than a threshold parameter $\epsilon \ll 1$. To minimize the amount of storage required by the tensor decomposition, it is desirable to make the $\epsilon$-threshold as large as possible. At the same time, $\epsilon$ must be sufficiently small to maintain the accuracy of the time-integration. In numerical experiments we study the interplay between the time-step and the $\epsilon$-threshold, with respect to both accuracy and storage requirements. We also evaluate how the computational effort scales with the number of sub-systems in the quantum system.

Related works include the development of tensor-based methods for solving differential equations using alternative decompositions, such as tree-tensor networks~\cite{CerLubWal-2021,CerKreSul-2025}, improvements in the accuracy of existing tensor evolution algorithms~\cite{Kusch2025}, and low-rank approaches for open quantum systems governed by the Lindblad equation~\cite{APPELO2025114036}. In parallel, ongoing research aims to improve the efficiency, accuracy, and scalability of quantum optimal control methods~\cite{Lee-2025,PETERSSON2025113712}.

The remainder of the paper is structured as follows. The governing equations are stated in Section~\ref{sec_quant-sim} and the conventional matrix-vector approach for solving Schr\"odinger's equation is described in Section~\ref{sec_mat-vec}. The tensor train decomposition and the time-integration algorithms TDVP and TDVP-2 are presented in Section~\ref{sec_tensor-train}. The BUG algorithm for Tucker tensors and tensor trains is outlined in Section~\ref{sec_bug}, followed by numerical experiments presented in Section~\ref{sec_num-exp}. Concluding remarks are given in Section~\ref{sec_conclusions}.

\section{Quantum simulations}\label{sec_quant-sim}

We consider the dynamics of a composite quantum system, consisting of $N$ sub-systems, where the state in each sub-system belongs to a  Hilbert (complex vector) space ${\cal H}_k$ of dimension $d_k\geq 1$, for $k\in[1,N]$. The composite state then belongs to the Hilbert space ${\cal H} = {\cal H}_1 \otimes \cdots \otimes {\cal H}_N = \mathbb{C}^{d_1\times \cdots \times d_N}$ of dimension $N_{tot}=\prod_{k=1}^{N} d_k$. 

The dynamics of the system is governed by Schr\"odinger's equation,
\begin{equation}\label{eq_schrodinger}
    |\dot{\psi}(t)\rangle = -i \hat{H}(t)|\psi(t)\rangle,
\end{equation}
where $\hat{H}(t)$ denotes the Hamiltonian operator and $|\psi(t)\rangle$ is the quantum state at time $t$. The Hamiltonian is a Hermitian operator, implying that Schr\"odinger's equation is time-reversible and that the norm of the state is conserved in time. Given an orthonormal basis in each sub-system $\{|\sigma_k\rangle\}$, the state is represented by
\begin{align}
    |\psi\rangle = \sum_{\sigma_1=1}^{d_1} \cdots \sum_{\sigma_N=1}^{d_N} \psi_{\sigma_1,\ldots,\sigma_N} | \sigma_1 \cdots \sigma_N\rangle,
\end{align}
where $\psi\in\mathbb{C}^{d_1\times \cdots\times d_N}$ is an an order-$N$ tensor. The Hamiltonian is a linear operator that maps a state in $\mathcal{H}$ onto another state in the same space, $\hat{H}:\, {\cal H} \to {\cal H}$. Formally, $\hat{H} \in {\cal B(H)} = {\cal H}\otimes {\cal H}^*$, where ${\cal H}^*$ is the dual space of ${\cal H}$. The Hamiltonian is represented by
\begin{align}\label{eq:Ham-gen}
    \hat{H} = \sum_{\sigma_1=1}^{d_1} \cdots \sum_{\sigma_N=1}^{d_N} \sum_{\sigma'_1=1}^{d_1} \cdots \sum_{\sigma'_N=1}^{d_N}
    {H}_{\sigma_1,\ldots,\sigma_N; \sigma'_1,\ldots,\sigma'_N}|\sigma_1\cdots\sigma_N\rangle\langle \sigma'_1\cdots \sigma'_N|,
\end{align}
where $H\in \mathbb{C}^{d_1^2\times \cdots \times d^2_N}$, noting that both vector and dual indices are subscripts in \eqref{eq:Ham-gen}. The action of the Hamiltonian operator on a state, $|\phi\rangle = \hat{H}|\psi\rangle$ follows by tensor contraction,
\begin{align}
    |\phi\rangle = \sum_{\sigma_1=1}^{d_1} \cdots \sum_{\sigma_N=1}^{d_N}\phi_{\sigma_1,\ldots,\sigma_N} |\sigma_1\cdots \sigma_N\rangle,\quad 
    \phi_{\sigma_1,\ldots,\sigma_N} = \sum_{\sigma'_1,\ldots,\sigma'_N} {H}_{\sigma_1,\ldots,\sigma_N;\sigma'_1,\ldots,\sigma'_N} \psi_{\sigma'_1,\ldots,\sigma'_N}.\label{eq_H-action}
\end{align}

In the following we focus on a Hamiltonian describing a quantum device with fixed transition frequencies and fixed couplings, including time-dependent control pulses. In a (global) frame rotating with angular frequency $\omega^{d}$, the Hamiltonian reads
\begin{align}
    \hat{H}(t) &= \hat{H}_s + \hat{H}_c(t), \label{eq:td-hamiltonian} \\
    \hat{H}_s &= \sum_{k=1}^{N} \left((\omega_k - \omega^{d}) \hat{a}_k^\dagger \hat{a}_k
    - \frac{\xi_k}{2} \hat{a}_k^\dagger \hat{a}_k^\dagger \hat{a}_k \hat{a}_k
    + \sum_{\ell>k} J_{k\ell} (\hat{a}_k^\dagger \hat{a}_\ell + \hat{a}_k \hat{a}_\ell^\dagger)\right), \label{eq:drift-hamiltonian} \\
    \hat{H}_c(t) &= \sum_{k = 1}^Np^k(t)(\hat{a}_k + \hat{a}_k^{\dag}) + iq^k(t)(\hat{a}_k - \hat{a}_k^{\dag}). \label{eq:control-hamiltonian}
\end{align}
Here, $\hat{a}_k$ is the lowering operator, $\omega_k$ and $\xi_k$ denote the transition frequency and anharmonicity, in sub-system $k\in[1,N]$. The coupling strength between sub-system $k$ and $\ell$ is $J_{k\ell}$. Note that this model includes the time-independent Heissenberg and Ising Hamiltonians as special cases. In the current study we assume that the sub-systems (qubits) are arranged in a linear chain and we consider the dynamics due to nearest neighbor interactions. 

\section{Matrix-vector formulation}\label{sec_mat-vec}

The quantum state can be represented in vectorized form,
\begin{align}
    \vec{\psib} := \text{vec}(\psi),\quad \vec{\psib}_\alpha = \psi_{k_1,\ldots,k_N},\quad \alpha = \widetilde{\mathbb{L}}(k_1,\ldots,k_N; d_1,\ldots,d_N),
\end{align}
where $\widetilde{\mathbb{L}}$ maps the tensor tuple indices $(k_1,\ldots,k_N)$, where $k_j\in[1,d_j]$, to a linear index using lexicographical ordering (last tuple index varies the fastest). The inverse of the index map is denoted $(k_1,\ldots,k_N) = \widetilde{\mathbb{T}}(\alpha; d_1,\ldots, d_N)$, see e.g.~Ballard and Kolda~\cite{Ballard_Kolda_2025}. In the following, the dimensions of the sub-systems will be suppressed to simplify the notation. 

In vectorized form, Schr\"odingers equation is
\begin{align} \label{eq:Schroedinger}
    \frac{d}{dt} \vec{\psib} = -i \mathbf{H} \vec{\psib},\quad t\geq 0,\quad \vec{\psib}(0)= \vec{\psib}_0,
\end{align}
where $\mathbf{H}\in\mathbb{C}^{N_{tot}\times N_{tot}}$ is the Hamiltonian matrix, with $N_{tot} = \prod_{k=1}^N d_k$. Corresponding to \eqref{eq_H-action}, the action of the Hamiltonian becomes
\begin{align}
    \vec{\phib}_\alpha = \sum_{\beta=1}^{N_{tot}} \mathbf{H}_{\alpha,\beta} \vec{\psib}_\beta,\quad
    \mathbf{H}_{\alpha,\beta} = H_{\sigma_1,\ldots,\sigma_N;\sigma'_1,\ldots,\sigma'_N},\quad \alpha\in[1,N_{tot}],
\end{align}
with $(\sigma_1,\ldots,\sigma_N) = \widetilde{\mathbb{T}}(\alpha)$ and $(\sigma'_1,\ldots,\sigma'_N) = \widetilde{\mathbb{T}}(\beta)$.
When evaluating the Hamiltonian model~\eqref{eq:drift-hamiltonian}-\eqref{eq:control-hamiltonian}, the lowering operator satisfies
\begin{align}\label{eq:lowering-k}
    \hat{a}_k = I_{d_1}\otimes \cdots \otimes I_{d_{k-1}} \otimes L_{[k]} \otimes I_{d_{k+1}}\otimes \cdots \otimes I_{d_N},
\end{align}
with 
\begin{align}
    L_{[k]} = \begin{bmatrix}
        0 & 1 &          &  &  \\
          & 0 & \sqrt{2} &  &  \\
          &   &  \ddots  & \ddots &  \\
          &   &          & 0      & \sqrt{d_{k}-1} \\
          &   &          &       & 0
    \end{bmatrix}\in\mathbb{R}^{d_k\times d_k}.
\end{align}
Note that the lexicographical ordering of $\vec{\psib}$ is consistent with defining $\otimes$ as the Kronecker product. 

When the Hamiltonian matrix is constant in time, the solution of Schr\"odinger's equation in vectorized form \eqref{eq:Schroedinger} can always be obtained by matrix exponentiation,
\begin{align}
    \vec{\psib}(t) = \exp(-i t\,\mathbf{H})\vec{\psib}_0. 
\end{align}
The computational complexity of matrix exponentiation is ${\cal{O}}(N_{tot}^3)$ operations, making this an efficient way of directly integrating Schr\"odinger's equation over long time domains, without time-stepping. However, this approach is only tractable when $N_{tot}$ is not too large. For larger system sizes, or when $\mathbf{H}$ depends on time, we must rely on numerical time-stepping to approximately solve Schr\"odinger's equation. For this purpose we define a grid in time, $t_k = k \delta$, for $k=0,1,2,\dots$, where $\delta>0$ is the time-step. We denote the numerical approximation of the state by $\vec{\psib}^k \approx \vec{\psib}(t_k)$. An unconditionally stable, norm-conserving and time-reversible scheme for integrating Schr\"odinger's equation is the implicit mid-point rule (IMR):
\begin{align}
    \vec{\psib}^{k+1} &= \vec{\psib}^k + \delta \vec{\gb}^{k+1/2},\quad k=0,1,2,\ldots,\quad \vec{\psib}^0 = \vec{\psib}(0),\label{eq:IMR-1}\\
    \vec{\gb}^{k+1/2} &= -\frac{i}{2} \mathbf{H}_{k+1/2} \bigl( \vec{\psib}^{k} + \vec{\psib}^{k+1} \bigr),\label{eq:IMR-2}
\end{align}
where $\mathbf{H}_{k+1/2} = \mathbf{H}(t_k+\delta/2)$.
By substituting $\vec{\psib}^{k+1}$ from \eqref{eq:IMR-1} into the right hand side of \eqref{eq:IMR-2} we arrive at a linear system for $\vec{\gb}^{k+1/2}$,
\begin{align}\label{eq_jacobi-iter}
    \bigl( I + \frac{i\delta}{2} \mathbf{H}_{k+1/2} \bigr)\vec{\gb}^{k+1/2} = \vec{\bb}^k,\quad \vec{\bb}^k = -i \mathbf{H}(t_{k+1/2}) \vec{\psib}^{k}.
\end{align}
Here, the linear system in \eqref{eq_jacobi-iter}
can be solved iteratively by, e.g., Jacobi's method. Note that iterative methods only require the action of the Hamiltonian to be evaluated. This approach requires less storage and is often significantly more efficient compared to first evaluating the Hamiltonian matrix and then applying it to a vector, especially when the Hamiltonian depends on time. 

Even though the IMR method is unconditionally stable, the time-step must chosen to be sufficiently small to give an accurate solution, and it is advisable to test the numerical accuracy with, e.g., Richardson extrapolation. As an alternative to IMR, Schr\"odinger's equation can also be solved by an explicit time-integration method, e.g., the classical RK-4 method. Similar to solving the linear system \eqref{eq_jacobi-iter} iteratively, explicit methods also only require the action of the Hamiltonian to be evaluated. For Schr\"odinger's equation, the stability region of the explicit method must include the imaginary axis and the time-step must be sufficiently small to ensure stability. The stability restrictions in an explicit method, and the accuracy requirements in an implicit method, are closely related because they both depend on the largest eigenvalue (in magnitude) of the Hamiltonian. Because of its unconditional stability, in this work we use the IMR method together with Richardson extrapolation for estimating the accuracy in the solution.

\section{Tensor trains}\label{sec_tensor-train}

A tensor train (TT), also called a matrix product state (MPS), represents an order-$N$ tensor as a product of $N$ order-3 tensors~\cite{PAECKEL2019167998}. Consider a tensor $\psi \in \mathbb{C}^{d_1 \times d_2 \times \cdots \times d_N}$. The elements in the tensor-train decomposition of $\psi$ can be written as 
\begin{equation}
        \psi_{\sigma_1, \ldots, \sigma_N} = 
        \sum_{m_0=1}^{b_0}\cdots
        \sum_{m_N=1}^{b_N} 
        M_{[1]}(m_0, \sigma_1, m_1)
        \cdots 
        M_{[N]}(m_{N-1}, \sigma_N, m_N),\quad \sigma_k \in [1,d_k].         \label{eq_MPS-el}
\end{equation}
The factors $M_{[k]} \in \mathbb{C}^{b_{k-1} \times d_k \times b_k}$ are order-3 tensors that are referred to as cores. Further, $b_{k}\geq 1$ (with $b_0=b_N=1$) is called the bond dimension between the cores $M_{[k]}$ and $M_{[k+1]}$. The indices $\sigma_k$ are called physical indices and $b_k$ are called virtual indices. For each fixed physical index, we may define a matrix
\[
M_{[k]}^{\sigma_k} = M_{[k]}(:,\sigma_k,:), \quad \sigma_k \in [1,d_k], \quad k\in[1,N].
\]
Then, each element in the tensor $\psi$ may be written as a matrix product,
\begin{equation}\label{eq:psi-MPS}
\psi_{\sigma_1, \ldots, \sigma_N} = M_{[1]}^{\sigma_1} M_{[2]}^{\sigma_2} \cdots M_{[N]}^{\sigma_N}.
\end{equation}
Note that $M_{[1]}^{\sigma_1}$ is a row vector because $b_0=1$ and $M_{[N]}^{\sigma_N}$ is a column vector because $b_N=1$. As a result, the matrix product on the right hand side of \eqref{eq:psi-MPS} evaluates to a scalar.
In the following, indices may be suppressed if they are clear from the context, e.g.,
\begin{align}
\psi = M_{[1]}\cdot M_{[2]} \cdots M_{[N]},
\end{align}
where a centered dot ($\cdot$) represents contraction of tensors over common indices. To further simplify the notation, it can be convenient to represent a tensor train in diagrammatical format, see Figure~\ref{fig:mps_diagram}.
\begin{figure}
    \centering
    \includegraphics[height=0.15\linewidth]{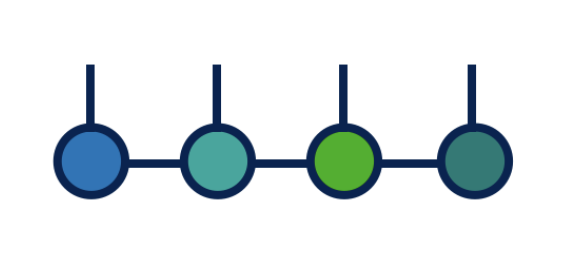}\hspace{5mm}
    \includegraphics[height=0.2\linewidth]{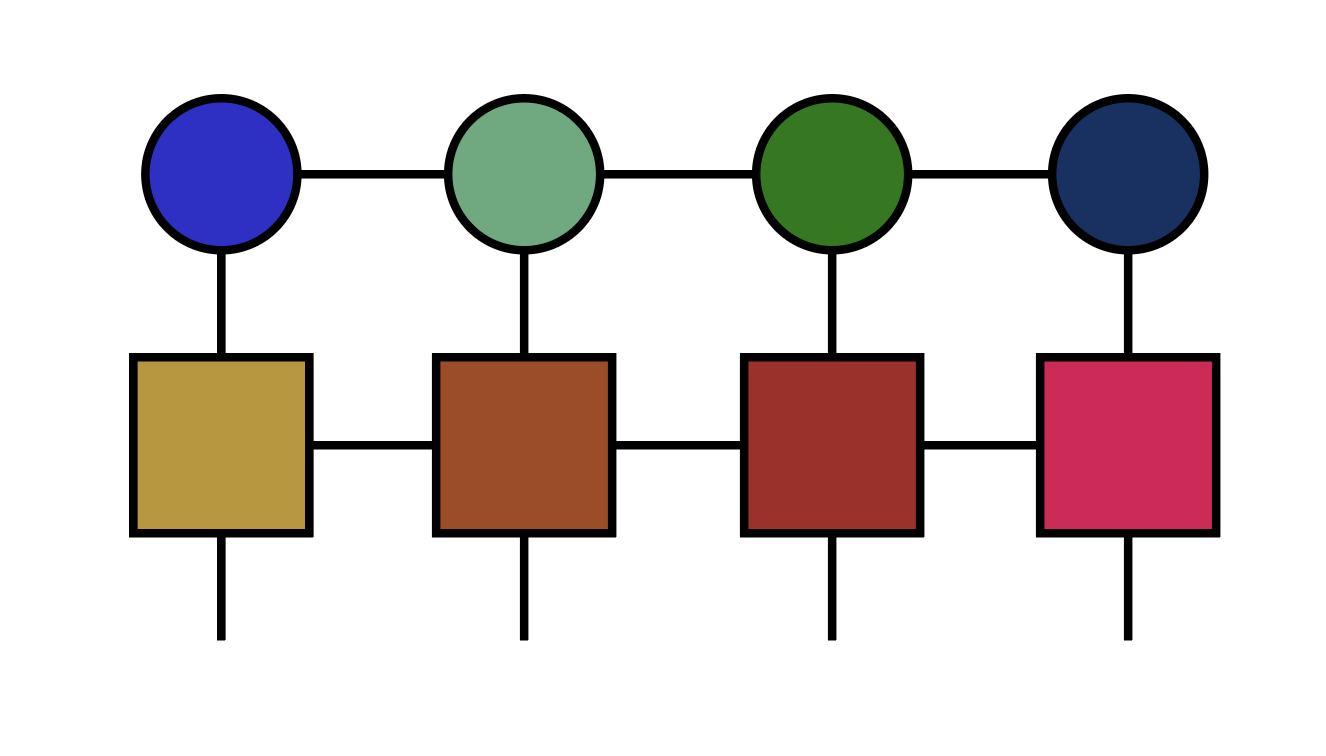}
    \caption{Left: A tensor diagram of a tensor train (MPS) of an order-4 tensor of size $d_1\times d_2\times d_3 \times d_4$, is represented by four cores (circles), each of order three. The vertical legs represent the physical indices and the horizontal legs connecting the cores correspond to contractions over virtual indices. The sizes of the first and last cores are $1\times d_1 \times b_1$ and $b_{N-1}\times d_N\times 1$, respectively; horizontal legs are suppressed on all singleton dimensions. Right: the contraction of an MPS by an order $4\times 4$ matrix product operator (MPO), represented by four (square) cores of order four.}
    \label{fig:mps_diagram}
\end{figure}
 
It is known that any tensor of order $N$ can be represented as a tensor train with $N$ cores, for example, using the TT-SVD algorithm~\cite{Ballard_Kolda_2025}. For a general tensor the bond dimensions must grow exponentially towards the center of the train, resulting in the same exponential growth of storage as in the vectorized representation. However, if the bond dimensions can be bounded by $b_k \leq B$, the total storage for $\psi$ becomes ${\cal O}(NB^2)$, which is linear in $N$, if $B$ can be bounded independently of $N$. 

The MPS representation of a given tensor is not unique because it is subject to gauge invariance. For example, we can insert the matrix identity $C^{-1}C = I$ between the cores $k$ and $k+1$ in \eqref{eq:psi-MPS}, resulting in the modified cores $\widetilde{M}_{[k]}^{\sigma_k} = M_{[k]}^{\sigma_k} C^{-1}$ and $\widetilde{M}_{[k+1]}^{\sigma_{k+1}} = C M_{[k+1]}^{\sigma_{k+1}}$, for any non-singular matrix $C\in\mathbb{C}^{b_k\times b_k}$. The gauge invariance can be exploited to improve the numerical stability of algorithms and simplify tensor contractions. In the following we will use two canonical forms of a tensor train:
\begin{align}
    \psi_{\sigma_1, \ldots, \sigma_N} = \begin{cases}
        A_{[1]}^{\sigma_1} \cdots A_{[j-1]}^{\sigma_{j-1}} M_{[j]}^{\sigma_j} B_{[j+1]}^{\sigma_{j+1}} \cdots B_{[N]}^{\sigma_N},& \text{(mixed canonical)},\\
        A_{[1]}^{\sigma_1} \cdots A_{[j]}^{\sigma_{j}} C_{[j]} B_{[j+1]}^{\sigma_{j+1}} \cdots B_{[N]}^{\sigma_N},& \text{(bond-centered canonical)}.
    \end{cases}
\end{align}
The canonical forms are defined by the left- and right-normalization conditions of the cores,
\begin{align}\label{eq:left-right-normal}
    \sum_{\sigma_k=1}^{d_k} \left(A_{[k]}^{\sigma_k}\right)^\dagger A_{[k]}^{\sigma_k} = I_{b_k},\ %
    \text{(left-normalized)},\quad
    \sum_{\sigma_k=1}^{d_k} B_{[k]}^{\sigma_k} \left(B_{[k]}^{\sigma_k}\right)^\dagger = I_{b_{k-1}},\ %
    \text{(right-normalized)}.
\end{align}
Here and in the following, left-normalized cores are denoted by $A_{[j]}$, right-normalized cores by $B_{[j]}$, and general cores by $M_{[j]}$. In the mixed canonical form, $M_{[j]}$ is called the orthogonality center. In the bond-centered canonical form, $C_{[j]}\in\mathbb{C}^{b_j\times b_j}$ is called the bond-matrix.

A quantum state on mixed canonical form can be written as
\begin{align}\label{eq:mixed-core-j}
    |\psi\rangle = \sum_{\sigma_j, m_{j-1}, m_{j}} |\psi^{L}_{j-1}(m_{j-1})\rangle \otimes \left(M_{[j]}^{\sigma_j}(m_{j-1},m_j)|\sigma_j\rangle\right) \otimes |\psi^{R}_{j+1}(m_j)\rangle,
\end{align}
where the left- and right-state vectors are defined by
\begin{align}
    |\psi^{L}_{j}(m_{j})\rangle &= \sum_{\sigma_1,\ldots,\sigma_{j}} A_{[1]}^{\sigma_1} \cdots A_{[j]}^{\sigma_{j}}(:,m_{j}) |\sigma_1\cdots \sigma_{j}\rangle,\quad
    m_{j}\in[1,b_{j}],\label{eq:psi-left}\\
    |\psi^{R}_{j+1}(m'_{j})\rangle &= \sum_{\sigma_{j+1},\ldots,\sigma_{N}} B_{[j+1]}^{\sigma_{j+1}}(m'_j,:) \cdots B_{[N]}^{\sigma_N} |\sigma_{j+1}\cdots\sigma_N\rangle,\quad
    m'_{j} \in[1,b_{j}].\label{eq:psi-right}
\end{align}
Each set of vectors has mutually orthogonal elements, as is made precise in the following lemma.
\begin{lemma}\label{lem:left-right-ortho}
    Assume that the basis vectors $\{|\sigma_k\rangle\}_{k=1}^N$ are mutually orthogonal and normalized.
    Let the set of left state vectors $\{|\psi^L_j(m_j)\rangle\}_{m_j=1}^{b_j}$ be defined by \eqref{eq:psi-left}, where the order-3 tensors $A_{[k]}$ satisfy the left-normalization condition in \eqref{eq:left-right-normal}. Then the left state vectors satisfy the orthogonality condition
    \begin{align}\label{eq:first-ortho}
        \langle \psi^L_j(\alpha) | \psi^L_j(\beta) \rangle = \delta_{\alpha,\beta},\quad \alpha,\, \beta \in [1,b_j].
    \end{align}
    Similarly, let the set of right state vectors $\{|\psi^R_j(m'_j)\rangle\}_{m'_j=1}^{b_j}$ be defined by \eqref{eq:psi-right}, where the order-3 tensors $B_{[k]}$ satisfy the right-normalization condition in \eqref{eq:left-right-normal}. Then the right state vectors satisfy the orthogonality condition
        \begin{align}\label{eq:second-ortho}
        \langle \psi^R_j(\alpha) | \psi^R_j(\beta) \rangle = \delta_{\alpha,\beta},\quad \alpha,\, \beta \in [1,b_j].
    \end{align}
\end{lemma}
\begin{proof} Follows directly from the left- and right-normalization conditions on $A_{[j]}$ and $B_{[k]}$.
\end{proof}

The orthogonality center in the mixed canonical form \eqref{eq:mixed-core-j} can be moved from $M_{[j]}$ to the bond between the cores $j$ and $j+1$. The first step is to factor $M_{[j]}$ into a left-normalized core and a square bond-matrix $C_{[j]}\in\mathbb{C}^{b_j\times b_j}$. Factorization is achieved by reshaping $M_{[j]}$ into a matrix $M_j^L \in\mathbb{C}^{b_{j-1} \cdot d_j \times b_j}$, followed by a singular-value decomposition $M^L_j =: U_j S_j V_j^\dagger$. A left-orthogonal core is constructed by inverse reshaping $U_j\in\mathbb{C}^{b_{j-1} \cdot d_j \times b_j}$ into the tensor $A_{[j]}\in\mathbb{C}^{b_{j-1}\times d_j \times b_j}$. Here, the left-orthogonality condition follows from $U_j^\dagger U_j = I_{b_j}$. The bond-matrix is defined by $C_{[j]} = S_j V_j^\dagger\in\mathbb{C}^{b_j\times b_j}$, resulting in
\begin{align}
    M_{[j]}^{\sigma_j}(m_{j-1},m_j) = \sum_{m'_j=1}^{b_j} A_{[j]}^{\sigma_j}(m_{j-1},m'_j) C_{[j]}(m'_j, m_j).
\end{align}
This decomposition results in the bond-centered canonical representation of the state,
\begin{align}\label{eq_schmidt-like}
    |\psi\rangle = \sum_{m_{j}, m'_{j}} C_{[j]}(m_{j},m'_j)|\psi^{L}_{j}(m_{j})\rangle \otimes |\psi^{R}_{j+1}(m'_j)\rangle.
\end{align}
The decomposition \eqref{eq_schmidt-like} is closely related to the Schmidt decomposition. The rank of the matrix $C_{[j]}$ is called the Schmidt rank and determines the entanglement between the states to the left and right of the bond between sites $j$ and $j+1$.

By contracting the bond-matrix $C_{[j]}$ with the core on the right, we can move the orthogonality center one step to the right,
\begin{align}\label{eq:ortho-right}
    M^{\sigma_{j+1}}_{[j+1]}(m_j,m_{j+1}) = \sum_{m'_j=1}^{b_j} C_{[j]}(m_j,m'_j)B_{[j+1]}^{\sigma_{j+1}}(m'_j,m_{j+1}).
\end{align}
This operation is an essential component in the TDVP algorithm (presented below) for solving Schroedinger's equation.

\subsection{Matrix product operators}

The tensor-train representation of a state can be generalized to represent the elements of the Hamiltonian~\eqref{eq:Ham-gen} as a matrix product operator (MPO),
\begin{align}\label{eq:MPO_notation}
    \hat{H}_{\sigma_1,\ldots,,\sigma_N;\sigma'_1,\ldots \sigma'_N} = 
    \sum_{w_0=1}^{b'_0} \cdots \sum_{w_N=1}^{b'_N} 
    W_{[1]}(w_0, \sigma_1,\sigma'_1, w_1) \cdots W_{[N]}(w_{N-1},\sigma_N, \sigma'_N,w_N).
\end{align}
Here, each core $W_{[k]}\in\mathbb{C}^{b'_{k-1}\times d_k\times d_k, b'_k}$ is an order 4 tensor. The first and last bond dimensions are $b'_0=b'_N=1$, and the interior bond dimensions satisfy $b'_k\geq 1$ for $k=1,\ldots,N-1$. The interior bond dimensions depend on the specific Hamiltonian model. For example, $b'_k=2$ suffices if the Hamiltonian only has local interactions within each sub-system. For all Hamiltonians with short-range interactions in one dimension, the required MPO bond dimension is small $(\approx 5)$ and is independent of the number of sub-systems~\cite{HubMccSch-2017}.

Similar to the MPS, for each pair of physical indices $(\sigma_k, \sigma'_k)$ we may define a matrix,
\begin{align}
    W_{[k]}^{\sigma_k,\sigma'_k} = W_{[k]}(:,\sigma_k, \sigma'_k,:)\in\mathbb{C}^{b'_{k-1}\times b'_k}, \quad \sigma_k, \sigma'_k \in [1,d_k],\quad k\in[1,N].
\end{align}
The MPO in matrix form allows the action of the Hamiltonian on a state to be expressed as an MPS. This operation is shown as a tensor diagram in Figure~\ref{fig:mps_diagram}. Let $|\psi\rangle$ be given as an MPS with general cores as in \eqref{eq:psi-MPS} and consider $|\phi\rangle = \hat{H}|\psi\rangle$. The resulting MPS for $|\phi\rangle$ can be obtained as the tensor product of the individual tensors, with elements
\begin{gather}\label{eq:H-psi}
    \phi_{\sigma_1,\ldots,\sigma_N} = \sum_{m_0,\ldots m_N}\sum_{w_0,\ldots,w_N} \widetilde{M}_{[1];(w_0m_0),(w_1m_1)}^{\sigma_1} \cdots \widetilde{M}_{[N];(w_{N-1}m_{N-1}),(w_N m_N)}^{\sigma_N}\\
    \widetilde{M}_{[k];(w_{k-1}m_{k-1}),(w_k m_k)}^{\sigma_k} = \sum_{\sigma'_k}
    W_{[k]}(w_{k-1}, \sigma_k,\sigma'_k, w_k)
    M_{[k]}(m_{k-1},\sigma'_k, m_k),\label{eq:H-psi-inflated}
\end{gather}
where the indexing of the core tensors in the resulting MPS is
\[
\widetilde{M}^{\sigma_k}_{[k];(w_{k-1}m_{k-1}),(w_k m_k)} := \widetilde{M}_{[k]}(m'_{k-1},\sigma_k, m'_k),\quad\text{with}\quad
m'_k = (w_k - 1)b_k + m_k.
\]
The state $|\phi\rangle$ is also in the MPS format, but with inflated bond dimensions: $b''_k = b'_k\cdot b_k$. Thus, repeated application of the Hamiltonian to a state quickly increases the storage requirements for the MPS. 
\begin{remark}
    The inflated MPS in \eqref{eq:H-psi}-\eqref{eq:H-psi-inflated} is never explicitly formed in the time-integration methods we describe below.
\end{remark}

\subsection{The time-dependent variational principle}

This section gives an overview of the time-dependent variational principle and the associated projection and contraction operators. These operators are used in the TDVP and TDVP-2 algorithms, which will be described in the following sections.

Variational approaches for solving the time-dependent Schroedinger equation date back to the works by Dirac~\cite{Dirac-1930}, Frenkel~\cite{Frenkel-1934},  McLachlan~\cite{McLachlan-1964}, and the time-dependent variational principle analyzed by Kramer~\& Saraceno~\cite{Kramer-1981} and Broeckhove et al.~\cite{Broeckhove-1988}.
In the variational approach, the evolution of the state is constrained to a manifold {\cal M} of matrix product states where the bond dimensions are fixed, resulting in the projected Schroedinger equation, 
\begin{equation}\label{eq_tdvp-eq}
    |\dot{\psi}\rangle = -i\hat{P}_{T, |\psi\rangle}\hat{H}(t)|\psi\rangle.
\end{equation}
Here, the action of the Hamiltonian on the state is projected onto the tangent plane of the manifold, where the tangent plane is centered at the current state $|\psi\rangle$. The state-dependent projector can be written as a sum of local projection operators~\cite{PAECKEL2019167998},
\begin{equation}\label{eq_tdvp-proj-sum}
    \hat{P}_{T, |\psi\rangle} = \sum_{j=1}^N \hat{P}^{L,\psi}_{j-1} \otimes \hat{I}_j \otimes \hat{P}^{R,\psi}_{j+1} - \sum_{i = 1}^{N - 1} \hat{P}^{L,\psi}_{j} \otimes \hat{P}^{R,\psi}_{j+1},
\end{equation}
where $\hat{P}^{L,\psi}_{j-1}$ and $\hat{P}^{R,\psi}_j$ are constructed from the gauge-fixed left- and right-normalized cores in the mixed-canonical \eqref{eq:mixed-core-j} and bond-centered \eqref{eq_schmidt-like} representations of the current state,
\begin{align}
    \hat{P}^{L,\psi}_{j-1} &= \sum_{m_{j-1}=1}^{d_{j-1}} |\psi^{L}_{j-1}(m_{j-1})\rangle \langle \psi^{L}_{j-1}(m_{j-1})| \otimes \hat{I}^R_{j},\\
    \hat{P}^{R,\psi}_{j+1} &= \hat{I}^L_{j} \otimes \sum_{m_j=1}^{d_j} |\psi^{R}_{j+1}(m_j)\rangle \langle \psi^{R}_{j+1}(m_j)|.
\end{align}
Here, $\hat{I}^R_j$ and $\hat{I}^L_{j}$ are the identity operators on sites $k\geq j$ and $k'\leq j$, respectively. 

When the state is on mixed canonical form with the orthogonality center at core $j$, it is straightforward to verify that $\hat{P}^{L,\psi}_{j-1}|\psi\rangle = |\psi\rangle$ and $\hat{P}^{R,\psi}_{j+1}|\psi\rangle = |\psi\rangle$. If we instead write the same state in the bond-centered canonical form, we easily see that $\hat{P}^{L,\psi}_j|\psi\rangle=|\psi\rangle$ and $\hat{P}^{R,\psi}_{j+1}|\psi\rangle=|\psi\rangle$. Since the same state can be written in either mixed canonical form (with the orthogonality center at any core), or in bond-centered form (where the bond-matrix is between any two consecutive cores), we conclude that every term in \eqref{eq_tdvp-proj-sum} is a projection operator. 

We proceed by defining a core-centered contraction operator for extracting the core when the state is on the mixed canonical form \eqref{eq:mixed-core-j} with the orthogonality center at site $j$. Note that $\langle \psi^{L}_{j-1}(m_{j-1})| \otimes \langle\sigma_j| \otimes \langle \psi^{R}_{j+1}(m_{j})|\psi\rangle = M_{[j]}^{\sigma_j}(m_{j-1}, m_j)$. We can therefore obtain all elements of $M_{[j]}^{\sigma_j}$ by defining the contraction operator as an outer product of left- and right-contraction operators,
\begin{align}
    \hat{Q}^{\sigma_j}_j = \begin{bmatrix}
       \langle \psi^{L}_{j-1}(1)| \\
       \vdots \\
       \langle \psi^{L}_{j-1}(b_{j-1})|  
    \end{bmatrix}\otimes \langle \sigma_j|\otimes 
    \Bigl[
        \langle \psi^{R}_{j+1}(1)|, \cdots, \langle \psi^{R}_{j+1}(b_{j})|
    \Bigr].
\end{align}
The core is then extracted by applying the contraction to the state,
\begin{align}\label{eq_QP}
    \hat{Q}^{\sigma_j}_j |\psi\rangle = 
    \begin{bmatrix}
        M_{[j]}^{\sigma_j}(1,1) & \cdots & M_{[j]}^{\sigma_j}(1,b_j) \\
        \vdots & \ddots & \vdots \\
        M_{[j]}^{\sigma_j}(b_{j-1},1) & \cdots & M_{[j]}^{\sigma_j}(b_{j-1},b_j)
    \end{bmatrix} = M_{[j]}^{\sigma_j},\quad\Leftrightarrow\quad
    \hat{Q}_j |\psi\rangle = M_{[j]}.
\end{align}
Another important property follows by applying the core-centered contraction operator to the $j$-th term in the first sum of \eqref{eq_tdvp-proj-sum},
\begin{multline}
    \hat{Q}_j^{\sigma_j}(m_{j-1}, m_j) 
    \left(\hat{P}^{L,\psi}_{j-1} \otimes \hat{I}_j \otimes \hat{P}^{R,\psi}_{j+1}\right) = \\
    \hat{Q}_j^{\sigma_j}(m_{j-1}, m_j)
    \sum_{m'_{j-1}, m'_j} \left(|\psi^L_{j-1}(m'_{j-1})\rangle \langle \psi^L_{j-1}(m'_{j-1})|\otimes \hat{I}_j \otimes 
    |\psi^R_{j+1}(m'_{j})\rangle \langle \psi^L_{j+1}(m'_{j})| \right) = \\
    \left(\langle \psi^L_{j-1}(m_{j-1})| \otimes \langle\sigma_j| \otimes \langle \psi^R_{j+1}(m_j)|\right) =
    \hat{Q}_j^{\sigma_j}(m_{j-1}, m_j),
\end{multline}
that is,
\begin{align}\label{eq_QP-contract}
   \hat{Q}_j 
    \left(\hat{P}^{L,\psi}_{j-1} \otimes \hat{I}_j \otimes \hat{P}^{R,\psi}_{j+1}\right) = \hat{Q}_j.
\end{align}

Assuming that the state is on the bond-centered canonical form \eqref{eq_schmidt-like}, with the bond-matrix between sites $j$ and $j+1$, we have $\langle \psi^{L}_{j}(m_{j})| \otimes \langle \psi^R_{j+1}(m'_j)|\psi\rangle = C_j(m_j,m'_j)$. The bond-centered contraction operator then follows from the outer product,
\begin{align}
    \hat{Q}^B_j = \begin{bmatrix}
       \langle \psi^{L}_{j}(1)| \\
       \vdots \\
       \langle \psi^{L}_{j}(b_{j})|    
    \end{bmatrix}\otimes
    \Bigl[
        \langle \psi^R_{j+1}(1)|, \cdots, \langle \psi^R_{j+1}(b_j)|
    \Bigr],\quad 
    \hat{Q}^B_j|\psi\rangle = C_{[j]}.
\end{align}
Furthermore, by applying the bond-centered contraction operator to the $j$-th term in the second sum of \eqref{eq_tdvp-proj-sum}, we get
\begin{align}\label{eq_QB}
    \hat{Q}^B_j
    \left(\hat{P}^{L,\psi}_{j}\otimes \hat{P}^{R,\psi}_{j+1}\right) =\hat{Q}^B_j. 
\end{align}

\subsection{The TDVP algorithm}

The TDVP algorithm~\cite{PAECKEL2019167998}, which is closely related to the projector-splitting method developed by Lubich et al.~\cite{Lub-Ose-Van-2015}, is based on the time-dependent variational principle. By inserting the explicit form of the projection operator \eqref{eq_tdvp-proj-sum} into the projected Schr\"odinger equation \eqref{eq_tdvp-eq}, we arrive at
\begin{equation}\label{eq:schroedinger-tdvp}
    |\dot{\psi}\rangle = -i\sum_{j=1}^N \left(\hat{P}^{L,\psi}_{j-1} \otimes \hat{I}_j \otimes \hat{P}^{R,\psi}_{j+1}\right)\hat{H}(t)|\psi\rangle 
    + i\sum_{j=1}^{N-1} \left(\hat{P}^{L,\psi}_{j} \otimes \hat{P}^{R,\psi}_{j+1}\right)\hat{H}(t)|\psi\rangle.
\end{equation}
In the TDVP algorithm, a splitting approach is employed in which each term on the right hand side of \eqref{eq:schroedinger-tdvp} is integrated individually and sequentially,
\begin{align}\label{eq_TDVP_sequence}
    &\begin{cases}
        |\dot{\psi}\rangle = -i \left(\hat{P}^{L,\psi}_{j-1} \otimes \hat{I}_{j} \otimes \hat{P}^{R,\psi}_{j+1}\right)\hat{H}(t)|\psi\rangle,  \\
        |\dot{\psi}\rangle = +i \left(\hat{P}^{L,\psi}_{j} \otimes \hat{P}^{R,\psi}_{j+1}\right)\hat{H}(t)|\psi\rangle,  
    \end{cases}
    \quad j=1,\ldots,N-1,\\
    &\hspace{4mm}|\dot{\psi}\rangle = -i\left(\hat{P}^{L,\psi}_{N-1} \otimes \hat{I}_N 
    \right)\hat{H}(t)|\psi\rangle.\label{eq_TDVP_last} 
\end{align}
Here we have defined $\hat{P}^{L,\psi}_{0}=1$ 
to allow for a uniform formulation of all terms. 

In TDVP, the state is calculated at discrete times, $t_k = k\delta$, where $\delta>0$ is the time-step and $k=1,\ldots,N_t$. Starting from the initial state written on mixed canonical form with the orthogonality center at the leftmost site ($j=1$), the state is evolved in two half-steps based on Strang splitting. In the first half-step, the individual cores are updated to time $t_0+\delta/2$ in a left-to-right sweep by solving \eqref{eq_TDVP_sequence}-\eqref{eq_TDVP_last}, resulting in a state on mixed canonical form with the orthogonality center at the rightmost site ($j=N$). This is followed by the second half-step in which the individual cores are evolved to time $t_0+\delta$ in a right-to-left sweep by solving by solving \eqref{eq_TDVP_sequence}-\eqref{eq_TDVP_last} in reversed order, resulting in a state in mixed canonical form with the orthogonality center back at the leftmost site. The time-stepping procedure is then repeated until the final time is reached.

By applying the core-centered contraction \eqref{eq_QP} to the first equation in \eqref{eq_TDVP_sequence}, we obtain the site-local Schroedinger equation for core $M_{[j]}$,
\begin{align}
    \dot{M}_{[j]}(t) = 
    -i\hat{Q}_j
    \left( \sum_{\sigma'_1,\ldots,\sigma'_N}
    \hat{H}^{\sigmab,\sigmab'}(t)
     A_{[1]}^{\sigma'_1} \cdots A_{[j-1]}^{\sigma'_{j-1}} M_{[j]}^{\sigma'_j}(t) B_{[j+1]}^{\sigma'_{j+1}} \cdots B_{[N]}^{\sigma'_N}|\sigmab'\rangle\right)
\end{align}
Here, $\hat{H}(t)$ is expressed in MPO format as in \eqref{eq:H-psi}.
On the right hand side, note that only the core $M_{[j]}$ depends on time; all other factors are fixed and can be pre-computed. We can therefore define an effective Hamiltonian acting on $M_{[j]}$ and write the site-local equation as
\begin{align}\label{eq_Mj-dot}
    \dot{M}_{[j]} = -i \hat{H}_j^{eff} \cdot {M}_{[j]},\quad
    \hat{H}_j^{eff} = L_{j-1}(t_0+\delta/2)\cdot W_{[j]}\cdot R_{j+1}(t_0),\quad t\in[t_0,t_0+\delta/2],
\end{align}
using $M_{[j]}(t_0)$ as initial condition. The effective Hamiltonian $\hat{H}^{eff}_j$ follows as a contraction between the contributions from core $W_{[j]}$ in the MPO, as well as the environments on the left ($L_{j-1}$) and right ($R_{j+1}$) of the orthogonality center. This decomposition is depicted as a tensor diagram in Figure \ref{fig:1site_effective} (left). We note that evaluating the action of the effective Hamiltonian only requires contractions of low-order tensors~\cite{PAECKEL2019167998}. Further, the right environments $(R_N, R_{N-1},\ldots R_2)$ can be recursively precomputed in a right-to-left sweep. The left environment $(L_{j-1})$ is updated based on $L_{j-2}$ during the left-to-right sweep. The resulting differential equation can be solved by IMR combined with an Jacobi iteration solving the associated linear system, cf.~\eqref{eq_jacobi-iter}.
\begin{figure}
    \centering
    \includegraphics[height=0.3\linewidth]{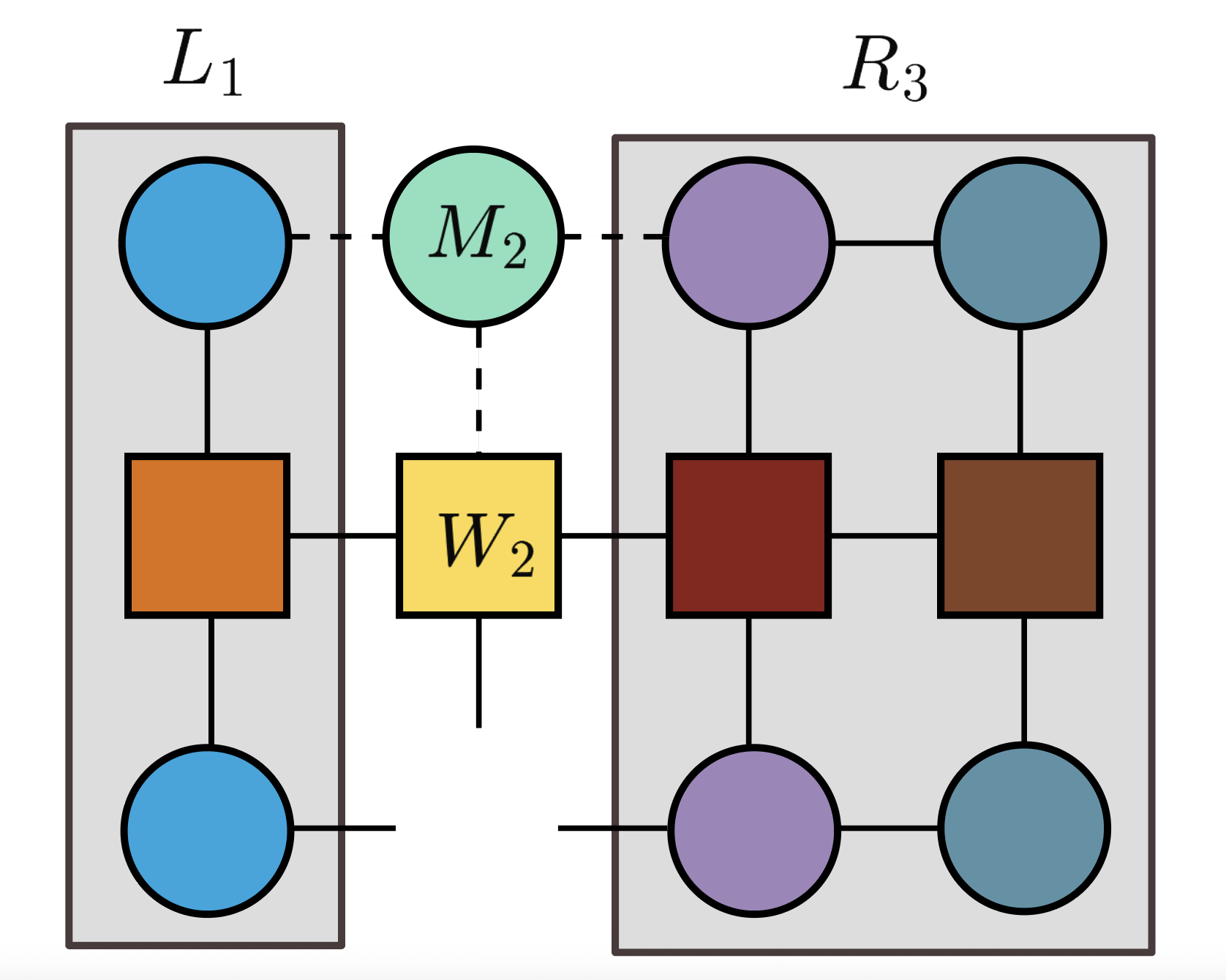} \hspace{10mm}
    \includegraphics[height=0.3\linewidth]{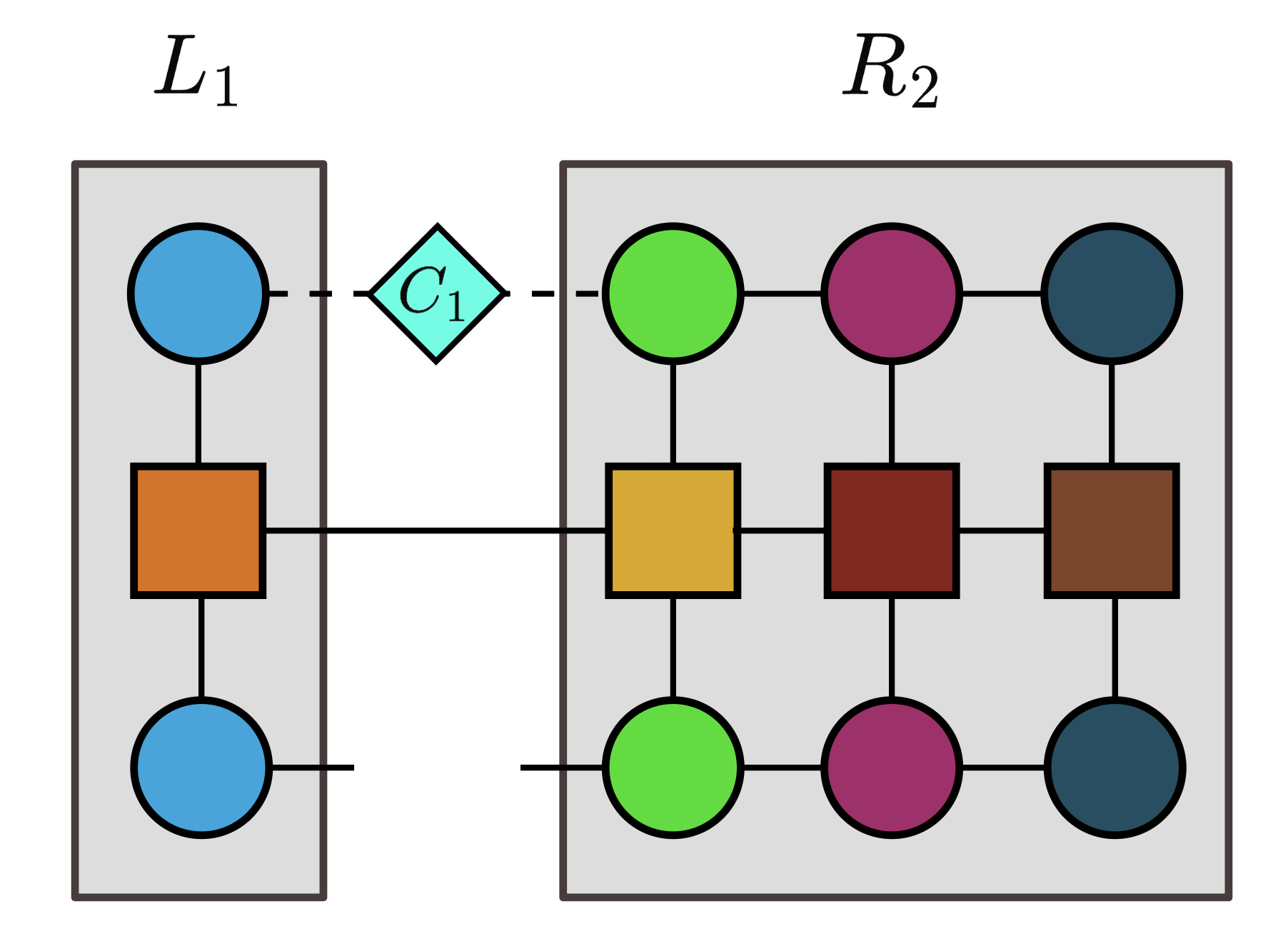}
    \caption{TDVP: evaluating the action of the effective Hamiltonian on the state in a 4-site system. The MPS of the state is in the top row, the MPO of the Hamiltonian in the middle row, and the resulting action is represented by the open legs in the bottom row. Left: the single-site Hamiltonian acting on the (green) core $M_{[2]}$, connected with dashed legs. Right: the bond-centered Hamiltonian acting on the (light blue) bond-matrix $C_{[1]}$, connected with dashed legs between the first and second core.}
    \label{fig:1site_effective}
\end{figure}

After the core $M_{[j]}$ has been evolved to time $t_0+\delta/2$, the orthogonality center is moved one step to the right as in \eqref{eq:ortho-right}, defining the left-normalized tensor $A_{[j]}(t_0+\delta/2)$ and the bond-matrix $C_{[j]}(t_0+\delta/2)$. The TDVP method then proceeds by solving the bond-matrix differential equation 
\begin{equation}\label{eq:schroedinger-tdvp-bwd}
    \dot{C}_{[j]} = +i\hat{K}^{eff}_j\cdot {C}_{[j]},\quad t_0+\delta \geq t\geq t_0,
\end{equation}
backwards in time. Again, the effective Hamiltonian $\hat{K}^{eff}_j = L_{j}(t_0+\delta/2)\cdot R_{j+1}(t_0)$ can be formed by contractions of low-order tensors as is illustrated in Figure~\ref{fig:1site_effective} (right). Given $C_{[j]}(t_0)$, the core at the next orthogonality center then follows from $M_{[j+1]}(t_0) := C_{[j]}(t_0) \cdot B_{[j+1]}(t_0)$. The left-to-right sweep is repeated until the rightmost site is reached, resulting in a mixed-canonical representation of the state at time $t_0+\delta/2$, with orthogonality center at site $N$. 

The second half-step of TDVP proceeds in a similar way but updates the cores in the opposite order, starting from the rightmost site and sweeping left, resulting in all cores being updated to time-level $t_k+\delta$. 

The TDVP algorithm conserves both the norm and the energy of the state during the time evolution~\cite{PAECKEL2019167998}. However, the accuracy of the solution depends on the bond dimensions that are fixed throughout the time evolution, i.e., they need to be assigned before the state evolution is known. Hence, the main difficulty with TDVP is to choose the bond dimensions such that they are large enough to accurately capture the state dynamics, but small enough to avoid making the problem computationally intractable. 

\subsection{The TDVP-2 algorithm}

The TDVP-2 algorithm overcomes the main shortcoming of the TDVP algorithm by adaptively changing the bond dimensions based on a threshold parameter, $\epsilon$, in a truncated SVD that will be described below.

Similar to TDVP, the TDVP-2 algorithm is based on a Strang splitting of the projected Schroedinger equation that evolves the state in two half-steps. To initialize, the state is first transformed on mixed canonical form with the orthogonality center at the leftmost site $(j=1)$. The first half-step in TDVP-2 proceeds by contracting $M_{[j]}$ and $B_{[j+1]}$ to form the merged order-4 tensor
\begin{align}
    M_{[j,j+1]}^{\sigma_j,\sigma_{j+1}}(m_{j-1},m_{j+1)} = \sum_{m'_{j}} M_{[j]}^{\sigma_j}(m_{j-1},m'_j) B_{[j+1]}^{\sigma_{j+1}}(m'_j, m_{j+1}),\quad    
    M_{[j,j+1]} = M_{[j]}\cdot B_{[j+1]}.
\end{align}
This tensor is then evolved forwards in time under an effective two-site Hamiltonian,
\begin{align}
    \dot{M}_{[j,j+1]} &= -i \hat{H}_{j,j+1}^{eff}\cdot M_{[j,j+1]},\quad  t_0\leq t\leq t_0+\delta/2,\\
    \hat{H}_{j,j+1}^{eff} &= L_{j-1}(t_0+\delta/2)\cdot W_{[j]}\cdot W_{[j+1]}\cdot R_{j+2}(t_0).
\end{align}
The action of the effective Hamiltonian on the merged core $M_{[j,j+1]}$ is illustrated in Figure~\ref{fig:2site_effective}.
\begin{figure}[tb]
    \centering
    \includegraphics[height=0.3\linewidth]{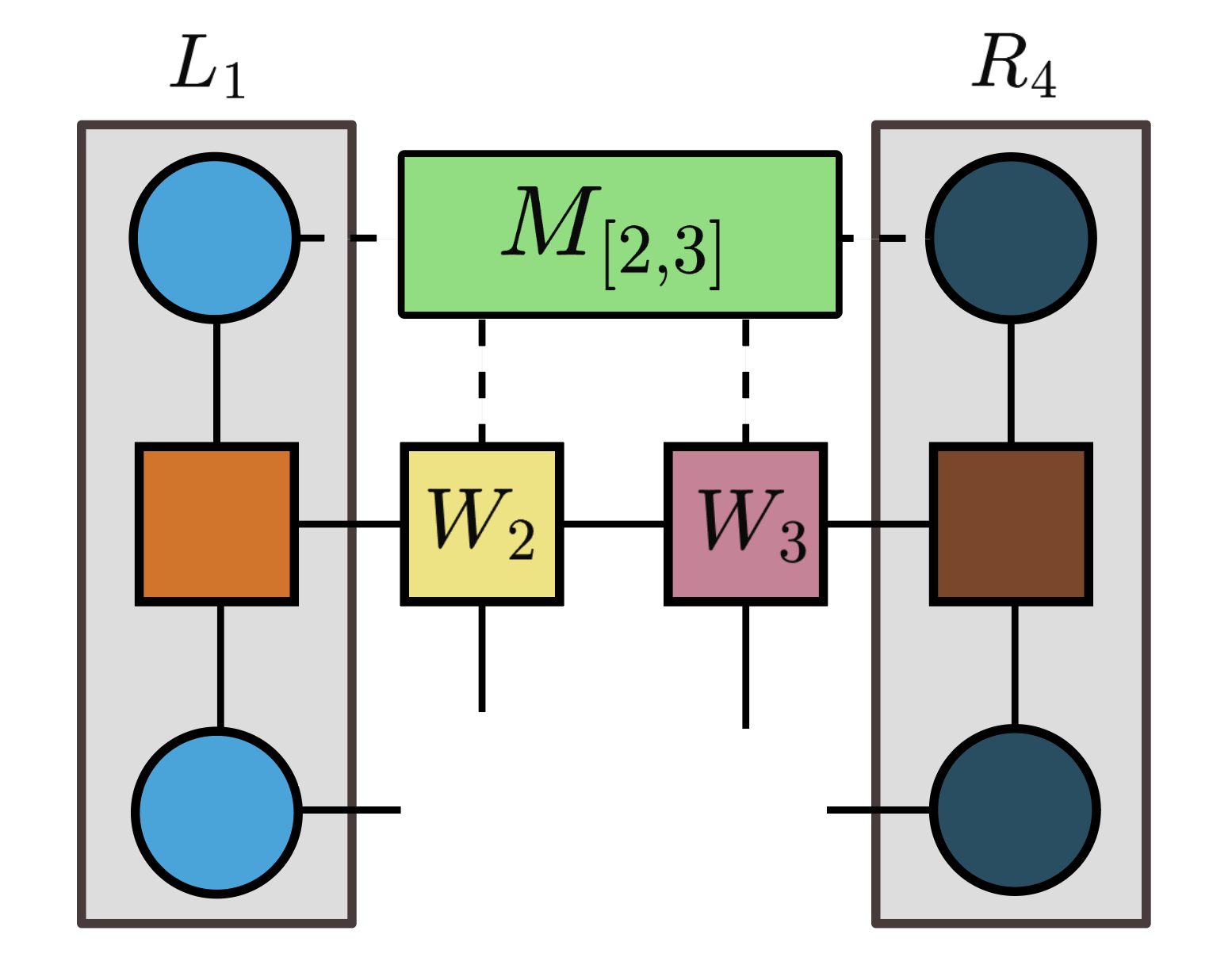} 
    \caption{TDVP-2: evaluating the action of the effective two-site Hamiltonian on the merged (green rectangular) core. The MPS of the state is in the top row, the MPO of the Hamiltonian in the middle row, and the resulting action is represented by the four open legs in the bottom row. }
    \label{fig:2site_effective}
\end{figure}
A truncated SVD is then applied to split the order-4 tensor back into two order-3 tensors, in which the threshold parameter ($\varepsilon$) controls the accuracy of the splitting and determines how the bond dimension $b_j$ needs to be modified (see Appendix~\ref{app:svd} for details). This results in
\begin{align}
    M_{[j,j+1]}(t_0+\delta/2) &\approx A_{[j]}(t_0+\delta/2) \cdot C_{[j]}(t_0+\delta/2) \cdot B_{[j+1]}(t_0+\delta/2),\\
    M_{[j]}(t_0+\delta/2) & := C_{[j]}(t_0+\delta/2) \cdot B_{[j+1]}(t_0+\delta/2),    
\end{align}
i.e., the orthogonality center is now at site $j+1$, but at time-level $t_0+\delta/2$. The next step is to evolve $M_{[j+1]}$ backwards in time by solving the single-site differential equation
\begin{align}
    \dot{M}_{[j+1]} = -i \hat{H}_{j+1}^{eff} \cdot {M}_{[j+1]},\quad
    \hat{H}_{j+1}^{eff} = L_{j}(t_0+\delta/2)\cdot W_{[j+1]}\cdot R_{j+2}(t_0),\quad t_0+\delta/2\geq t\geq t_0,
\end{align}    
resulting in $M_{j+1}(t_0)$. Note that this equation is identical to \eqref{eq_Mj-dot} in TDVP, but it is evolved backwards in time. The merged core, $M_{[j+1,j+2]}(t_0)$, can then formed and the process is repeated until the rightmost site is reached. Similar to TDVP, the second half-step updates the cores in the opposite order in a right-to-left sweep resulting in an approximation of the state at time $t_0+\delta$.

\section{The basis update and Galerkin (BUG) algorithm}\label{sec_bug}

The BUG algorithm is a time integration method that can be used to solve Schr\"odinger's equation when the state is represented by a tensor decomposition. The most basic case occurs for a tensor of order 2, i.e., a matrix differential equation. The BUG algorithm is based on a splitting approach that avoids stiffness due to small singular values in the evolved matrix~\cite{Ceruti2022}, making explicit time-integration methods tractable and improving the convergence of iterative solvers within an implicit method. The basic BUG algorithm has been generalized to evolve tensor differential equations where the state is represented by an arbitrary hierarchical tree tensor decomposition, see Lubich et al.~\cite{BUG_TreeTensorNetworks}. While this algorithm also applies to Tucker tensors and tensor trains as special cases, it is non-trivial to decipher a practical implementation from this very general description. In this section we first outline the rank-adaptive BUG algorithms for matrices and Tucker tensors, followed by the BUG algorithm for tensor trains (MPS-BUG).

\subsection{BUG for matrices and Tucker tensors}

The basic BUG algorithm applies to matrix differential equations. Consider the order-2 tensor (matrix) $T \in \mathbb{C}^{d_1 \times d_2}$ governed by the differential equation,
\begin{align}
    \dot{T} = F(t, T),\quad t > 0,\quad T(0) = T_0,\label{eq_Mat-DiffFQ}
\end{align}
where the function $F\in \{\mathbb{R}\times \mathbb{C}^{d_1 \times d_2} \to \mathbb{C}^{d_1 \times d_2}\}$ is assumed to be sufficiently regular.

The order-2 tensor is represented in the form $T=USV^\dagger$, with $U^\dagger U = I$ and $V^\dagger V = I$. This resembles a singular value decomposition, except that the matrix $S$ is not assumed to be diagonal. As before, we denote the time-step by $\delta > 0$. Starting from the $T_0=T(t_0)$, factored as $T_0 = U_0 S_0 V_0^\dagger$, one time-step of the rank-adaptive BUG scheme calculates $T_1\approx T(t_1)$, with $t_1=t_0+\delta$, as described in Algorithm~\ref{alg:BUG-matrix}. The algorithm is repeated until the desired final time is reached.
\begin{algorithm} 
  \caption{One time-step with rank-adaptive BUG for matrix differential equations.} \label{alg:BUG-matrix}
  \begin{algorithmic}[1]
    \State \textbf{Inputs:} Initial state at time $t_0$: $T_0 = U_0 S_0 V_0^\dagger$, SVD threshold $\epsilon\geq 0$, time-step $\delta$.
    \State \textbf{Outputs:} State at time $t_1 = t_0+\delta$: $T_1 = U_1 S_1 V_1^\dagger$.
    \State \textbf{K-step:} Integrate the differential equation:
    \[
    \dot{K}(t) = F\bigl(t, K(t)V_0^{\dag}\bigr)\,V_0,\quad t_0<t\leq t_1,\quad K(t_0) = U_0 S_0.
    \]
    \State \textbf{QR-step \#1:} QR-factorize the concatenated matrix [$K(t_1) | U_0] =: \hat{U}_1 R_1$.
    \State \textbf{L-step:} Integrate the differential equation: 
    \[
        \dot{L}(t) = \Bigl(F\bigl(t, U_0L(t)^{\dag}\bigr)\Bigr)^{\dag} U_0, \quad t_0<t\leq t_1,\quad L(t_0) = V_0 S_0^{\dag}.
    \]
    \State \textbf{QR-step \#2:} QR-factorize the concatenated matrix $[L(t_1) | V_0] =: \hat{V}_1 R_2$.
    \State \textbf{S-step:} Integrate the differential equation: 
    \[
        \dot{\hat{S}}(t) = \hat{U}_1^{\dag} F\bigl(t, \hat{U}_1 \hat{S}(t) \hat{V}_1^{\dag}\bigr)\,\hat{V}_1,\quad t_0<t\leq t_1,\quad \hat{S}(t_0) = \hat{U}_1^{\dag} U_0 S_0 V_0^{\dag} \hat{V}_1.
    \]
    \State \textbf{Truncate bond dimension:} Calculate $\epsilon$-truncated SVD : 
    $\hat{S}(t_1) \approx: U_1 S_1 V_1^\dagger$.
    \State \textbf{State at time $t_1$:} $T_1 = U_1 S_1 V_1^{\dag}$.
  \end{algorithmic}
\end{algorithm}


We proceed by discussing the BUG algorithm for Tucker tensors.
Any tensor $Y \in \mathbb{C}^{d_1 \times ... \times d_N}$ can be represented in Tucker format as a factorization of one order-$N$ tensor $G$ and $N$ matrices $U_i$~\cite{Ballard_Kolda_2025},
\begin{align}
        Y(i_1, \ldots, i_N) = \sum_{j_1 = 1}^{b_1}\cdots \sum_{j_N = 1}^{b_N}G(j_1, \ldots, j_N)U_1(i_1, j_1)\cdots U_N(i_N, j_N).\label{eq_Tucker-el}
\end{align}
Here, $G \in \mathbb{C}^{b_1\times ... \times b_N}$ is called the core, and $U_k \in \mathbb{C}^{d_k \times b_k}$ for $k\in[1,N]$, are the factor matrices. Note that the Tucker format can only give a compressed representation of $Y$ when the bond dimensions are smaller than the physical dimensions, i.e., $b_k < d_k$ for $k\in[1,N]$. A Tucker-tensor decomposition can be written as a sequence of Tensor-Times-Matrix (TTM) products~\cite{Ballard_Kolda_2025}, denoted as
\begin{align}\label{eq_Tucker-TTM}
        Y = G \times_1 U_1 \cdots \times_N U_N = G \times_{j=1}^N U_j,\quad \text{or}\quad
        Y = \ldbrack G; U_1,\ldots, U_N \rdbrack.
\end{align}

Let the Tucker tensor defined above satisfy the tensor differential equation, 
\begin{equation}\label{eq_Tensor-DiffEq}
    \dot{Y}(t) = F(t, Y(t)),\quad t_0 < t \leq T,\quad 
    Y(t_0) = \ldbrack G(t_0); U_1(t_0), \ldots, U_N(t_0) \rdbrack =: Y_0.
\end{equation}
As before, the function $F\in\{ \mathbb{R} \times \mathbb{C}^{d_1 \times ... \times d_N} \to \mathbb{C}^{d_1 \times ... \times d_N} \}$ is assumed to be sufficiently regular. The rank-adaptive BUG algorithm~\cite{CerKusLub-2022} computes a Tucker tensor $Y_1 \approx Y(t_1)$, where 
\begin{equation}\label{eq_bug-tucker-step}
    Y_1 = \ldbrack G(t_1), U_1(t_1), \ldots, U_N(t_1) \rdbrack.
\end{equation}
Here, the $K$ and $L$ steps in Algorithm~\ref{alg:BUG-matrix} are replaced by $N$ sub-steps for solving the corresponding matricised differential equations. Here we only remark on how the algorithm is used for solving Schr\"odinger's equation. In this case we need to evaluate $F(t,Y) = -i\hat{H}(t)Y$. Using lexicographical ordering,
\[
\text{vec}(Y) = (U_1\otimes \cdots \otimes U_N) \text{vec}(G).
\]
In general the Hamiltonian can be written as $\hat{H} = \sum_j \hat{H}_j$, where each term satisfies,
\[
\hat{H}_j = S_1 \otimes \cdots \otimes S_N,\quad
\hat{H}_j\text{vec}(Y) = \left(S_1 U_1\otimes \cdots \otimes S_N U_N\right)\text{vec}(G).
\]
The action of each term in the Hamiltonian on a Tucker tensor thus results in another Tucker tensor, with the same core and same dimensions of the factor matrices. Such Tucker tensors can be added efficiently, e.g., by calculating a site-local basis for the factor matrices at each site. We conclude that the action of the Hamiltonian on a Tucker tensor can be performed efficiently through site-local operations. For further details, we refer to Ceruti et al.~\cite{CerKusLub-2022}.

\subsection{BUG for tensor trains}\label{sec:bug_tt}

Compared to the TDVP-2 method~\cite{PAECKEL2019167998} the BUG algorithm for tensor trains has one significant advantage: each sub-step in BUG only involves forwards in time evolution. While this may not be important for the Schr\"odinger equation, which is time-reversible, the backwards sub-step within TDVP-2 could lead to time-stepping instabilities for dissipative differential equations. For example, in the Monte-Carlo wave function method where the evolution is governed by Schrödinger's equation with a non-Hermitian effective Hamiltonian~\cite{Molmer-93}.

In contrast to the TDVP and TDVP-2 algorithms, where the orthogonality center is moved from one core to the next, the MPS-BUG method keeps the orthogonality center fixed throughout the time-stepping, typically somewhere near the middle of the train. In the following, assume the state to be given as an MPS at time $t_0$, with the orthogonality center at site $j_c$,
\begin{equation}
    \psi_{\sigma_1,\ldots,\sigma_N}(t_0) =
    A_{[1]}^{\sigma_1}(t_0)\cdots A_{[j_c - 1]}^{\sigma_{j_c-1}}(t_0)M_{[j_c]}^{\sigma_{j_c}}(t_0) B_{[j_c + 1]}^{\sigma_{j_c+1}}(t_0) \ldots B_{[N]}^{\sigma_{N}}(t_0),
\end{equation}
where the tensor cores have the sizes
\begin{align}
    A_{[k]}(t_0) \in \mathbb{C}^{b_{k - 1} \times d_k \times b_k}, \quad 
    M_{[j]}(t_0) \in \mathbb{C}^{b_{j -1}\times d_j \times b_j}, \quad 
    B_{[m]}(t_0) \in \mathbb{C}^{b_{m - 1} \times d_m \times b_{m}},
\end{align}
for $k \in [1, j_c -1]$ and $m \in [j_c + 1, N]$, with $b_0=b_N=1$. 

In each timestep of the MPS-BUG algorithm, the cores to the left of the orthogonality center, $A^{\sigma_k}_{[k]}$ for $k\in[1,j_c-1]$, are updated in a left-to-right sweep, while the cores to the right of the orthogonality center, $B^{\sigma_m}_{[m]}$ for $m \in[j_c+1,N]$, are updated in a right-to-left sweep. These sweeps are independent of each another and can be performed concurrently. The core at the orthogonality center, $M^{\sigma_{j_c}}_{[j_c]}$, is updated after all $A_{[k]}(t_1)$ and $B_{[m]}(t_1)$ have been calculated. During these operations, all bond dimensions are temporarily doubled to account for potential increases in entanglement. At the end of each timestep, the bond dimensions are re-evaluated based on a truncated singular value decomposition of each core, where singular values smaller than $\epsilon$ are removed.

Evolving the state by one time-step results in another MPS,
\begin{align}\label{eq_psi-t1}
    \psi_{\sigma_1,\ldots,\sigma_N}(t_1) \approx
    A_{[1]}^{\sigma_1}(t_1)\cdots A_{[j_c - 1]}^{\sigma_{j_c-1}}(t_1)
    M_{[j_c]}^{\sigma_{j_c}}(t_1)
    B_{[j_c + 1]}^{\sigma_{j_c+1}}(t_1) \cdots B_{[N]}^{\sigma_{N}}(t_1),\quad t_1=t_0+\delta,
\end{align}
where the cores temporarily are assigned doubled bond dimensions, 
\begin{align}\label{eq_psi-t1-sizes}
    A_{[k]}(t_1) \in \mathbb{C}^{2b_{k - 1} \times d_k \times 2b_k}, \quad
    M_{[j_c]}(t_1) \in \mathbb{C}^{2b_{j_c -1}\times d_j \times 2b_{j_c}}, \quad 
    B_{[m]}(t_1) \in \mathbb{C}^{2b_{m - 1} \times d_m \times 2b_{m}},
\end{align}
for $k \in [2, j_c -1]$ and $m \in [j_c + 1, N-1]$. Note that the bond dimensions at both ends of the tensor train are not doubled. Instead, $A_{[1]}(t_1) \in \mathbb{C}^{1 \times d_1 \times 2b_1}$ and $B_{[N]}(t_1) \in \mathbb{C}^{2b_{N - 1} \times d_N \times 1}$. 

Similar to the site-local differential equation in the TDVP algorithm \eqref{eq_Mj-dot}, the sites to the left and right of the orthogonality center are evolved by sequentially by solving local Schr\"odinger equations based on effective Hamiltonians. Sites to the left of the orthogonality center are updated by solving the site-local Schr\"odinger equation:
\begin{align}\label{eq_bug_tt_left}
    \dot{\widetilde{A}}_{[k]} = -i\hat{H}^{eff}_{k} \cdot \widetilde{A}_{[k]}, \quad \hat{H}_k^{eff} = L_{k -1}(t_1)\cdot W_{[k]} \cdot R_{k + 1}(t_0),\quad k \in [1, j_c - 1],
\end{align}
where $\widetilde{A}_{[k]} \in\mathbb{C}^{2b_{k-1}\times d_k\times b_k}$ (except for $k=1$ where the first mode always has dimension 1).
A tensor diagram for evaluating the action of the effective Hamiltonian is shown in Figure~\ref{fig:h_eff_bug}. 
\begin{figure}
    \centering
    \includegraphics[width=0.45\linewidth]{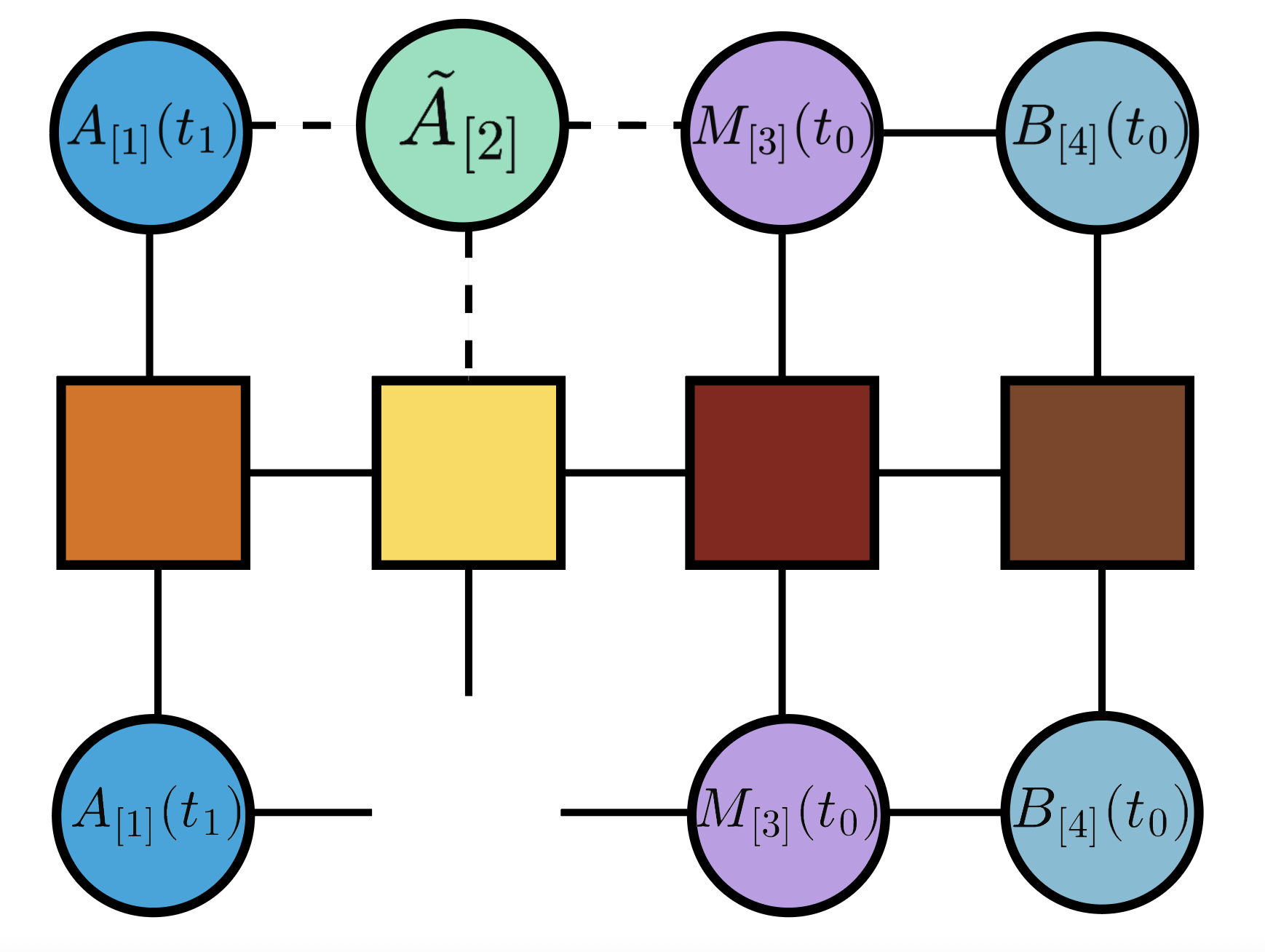}
    \caption{MPS-BUG: application of $\hat{H}^{eff}_2$ to $\widetilde{A}_{[2]}$ (green, top row), giving the slope $\hat{H}^{eff}_2{\widetilde{A}}_{[2]}$ (3 solid legs into empty space, bottom row).}
    \label{fig:h_eff_bug}
\end{figure}
In a corresponding way, sites to the right of the orthogonality center are updated by solving the site-local differential equation: 
\begin{align}\label{eq_bug_tt_right}
    &\dot{\widetilde{B}}_{[m]} = -i\hat{H}_m^{eff}\cdot \widetilde{B}_{[m]}, \quad 
    \hat{H}_m^{eff} = L_{m - 1}(t_0)\cdot W_{[m]}\cdot R_{m + 1}(t_1), \quad m \in [j_c + 1, N],
\end{align}
where $\widetilde{B}_{[m]} \in\mathbb{C}^{b_{m-1}\times d_m\times 2b_m}$ (except for $m=N$ where the third mode always has dimension 1). 

Before these site-local equation can be solved, we note two differences from the TDVP algorithm. First, the bond dimensions of the initial conditions for $\widetilde{A}_{[k]}(t_0)$ and $\widetilde{B}_{[m]}(t_0)$ are different from the bond dimensions of $A_{[k]}(t_0)$ and $B_{[m]}(t_0)$, respectively. Secondly, the bond dimension of $\widetilde{A}_{[k]}(t_1)$ is different from ${A}_{[k]}(t_1)$, and the bond dimension of $\widetilde{B}_{[m]}(t_1)$ is different from ${B}_{[m]}(t_1)$. The algorithms for accurately consolidating the tensor dimensions is beyond the scope of this presentation, but will be discussed in detail in a forthcoming paper by Sulz~\cite{Sulz-2026}.

Once $A_{[k]}(t_1)$ and $B_{[m]}(t_1)$ have been calculated for all $k$ and $m$, the site at the orthogonality center is updated by solving a third site-local Schr\"odinger equation: 
\begin{align}\label{eq_bug-tt-center}
    \dot{\widetilde{M}}_{[j_c]} &= -i\hat{H}^{eff}_{j_c} \cdot \widetilde{M}_{[j_c]},\quad \hat{H}_{j_c}^{eff} = L_{{j_c} - 1}(t_1) \cdot W_{[j_c]} \cdot R_{j_c + 1}(t_1).
\end{align}

After all sites of the tensor train have been updated, the state at time $t_1$ is represented in the form~\eqref{eq_psi-t1}-\eqref{eq_psi-t1-sizes}, where all bond dimension have been doubled. To limit the storage requirements while maintaining the prescribed solution accuracy, a truncated singular value decomposition is then performed on each core, in which all singular values smaller than $\epsilon$ are set to zero, thus allowing the bond dimensions to be reduced accordingly.

\section{Numerical Examples}\label{sec_num-exp}

In this section we numerically evaluate the computational performance for solving Schr\"odinger's equation using the Tucker tensor and tensor train (MPS) decompositions, and compare them to the conventional matrix-vector approach.  
In the matrix-vector formulation, Schr\"odinger's equation can be solved by matrix exponentiation when the Hamiltonian is time-independent and the dimension of the state vector is sufficiently small. For time-dependent Hamiltonians, Schr\"odinger's equation must be solved numerically; here we use the C++ code Quandary~\cite{QuandaryUser} that implements the approach in Section~\ref{sec_mat-vec}. The tensor train and Tucker tensor algorithms from Sections~\ref{sec_tensor-train} and~\ref{sec_bug} are implemented in the Julia language, based on the iTensor package~\cite{iTensor}. All calculations are performed on a MacBook Pro with an Apple M4 chip.

To evaluate the effectiveness of the TDVP-2 and the rank-adaptive BUG methods, we apply the algorithms to both a time-independent and a time-dependent Hamiltonian model. In the time-independent case, we study the transverse Ising model. In the time-dependent case, we consider a simplified model of a quantum device subject to time-dependent control pulses.

\subsection{Time-independent Hamiltonians}

We start by considering the one-dimensional transverse-field Ising model for $N$ qubits, with Hamiltonian 
\begin{equation}
    H = -J\sum_{j=1}^{N-1} S^z_j S^z_{j+1} - g \sum_{j=1}^{N} S^x_j,
    \label{eq:constant-hamiltonian}
\end{equation}
which has a matrix product operator (MPO) representation with bond dimension $3$, see~\cite{10.21468/SciPostPhys.17.5.135}. The operators $S_j^z$ and $S_j^x$ are defined by,
\begin{align*}
    S_j^z &= I^{\otimes (j-1)} \otimes S^z \otimes I^{\otimes (N-j)},\\
    S_j^x &= I^{\otimes (j-1)} \otimes S^x \otimes I^{\otimes (N-j)},
\end{align*}
where
\[
    S^z =
    \tfrac{1}{2}\begin{bmatrix}
        1 & 0 \\
        0 & -1
    \end{bmatrix},
    \qquad
    S^x =
    \tfrac{1}{2}\begin{bmatrix}
        0 & 1 \\
        1 & 0
    \end{bmatrix}.
\]

Unless otherwise noted, all numerical simulations in this section use a random product state as the initial state $\psi(0)$, the error in the states is measured in $L_2$ norm at the final time $T=10.0$, and the exact state is obtained via matrix exponentiation. Without truncation of the SVD ($\epsilon=0$), our implementations of the TDVP-2, Tucker-BUG, and MPS-BUG algorithms have all been verified to be second order accurate in the time-step.

A major advantage of tensor decompositions is that adding sub-systems (qubits) to the composite quantum system does not necessarily lead to an exponentially increasing computational cost. This is demonstrated in Figure~\ref{fig:method_convergence}, which shows run times and bond dimensions for evolving the Ising model to $T=5.0$, for increasing numbers of sub-systems ($N$). 
\begin{figure}
    \centering
    \includegraphics[width=0.49\linewidth]{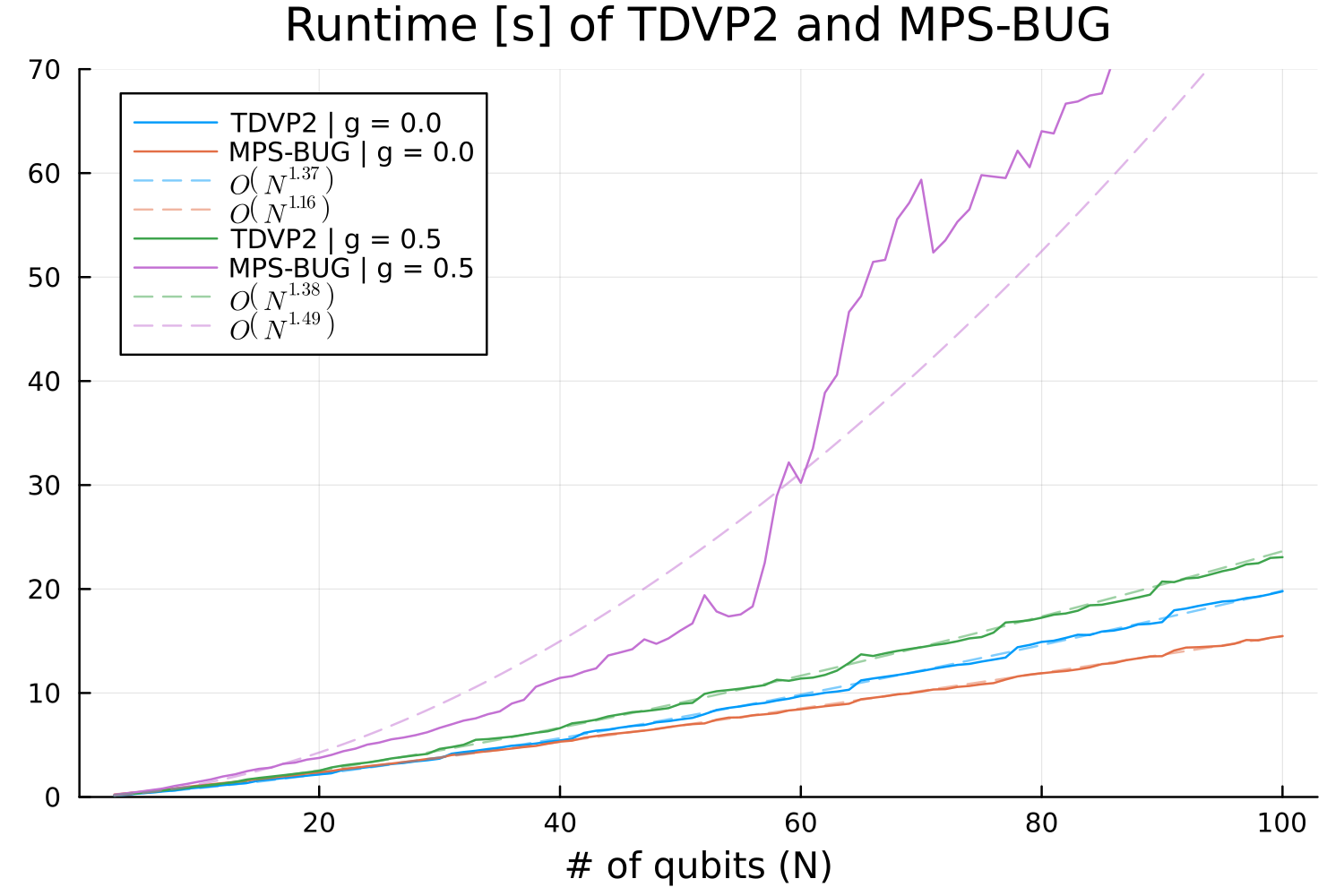}
    \hspace{1mm}
    \includegraphics[width=0.49\linewidth]{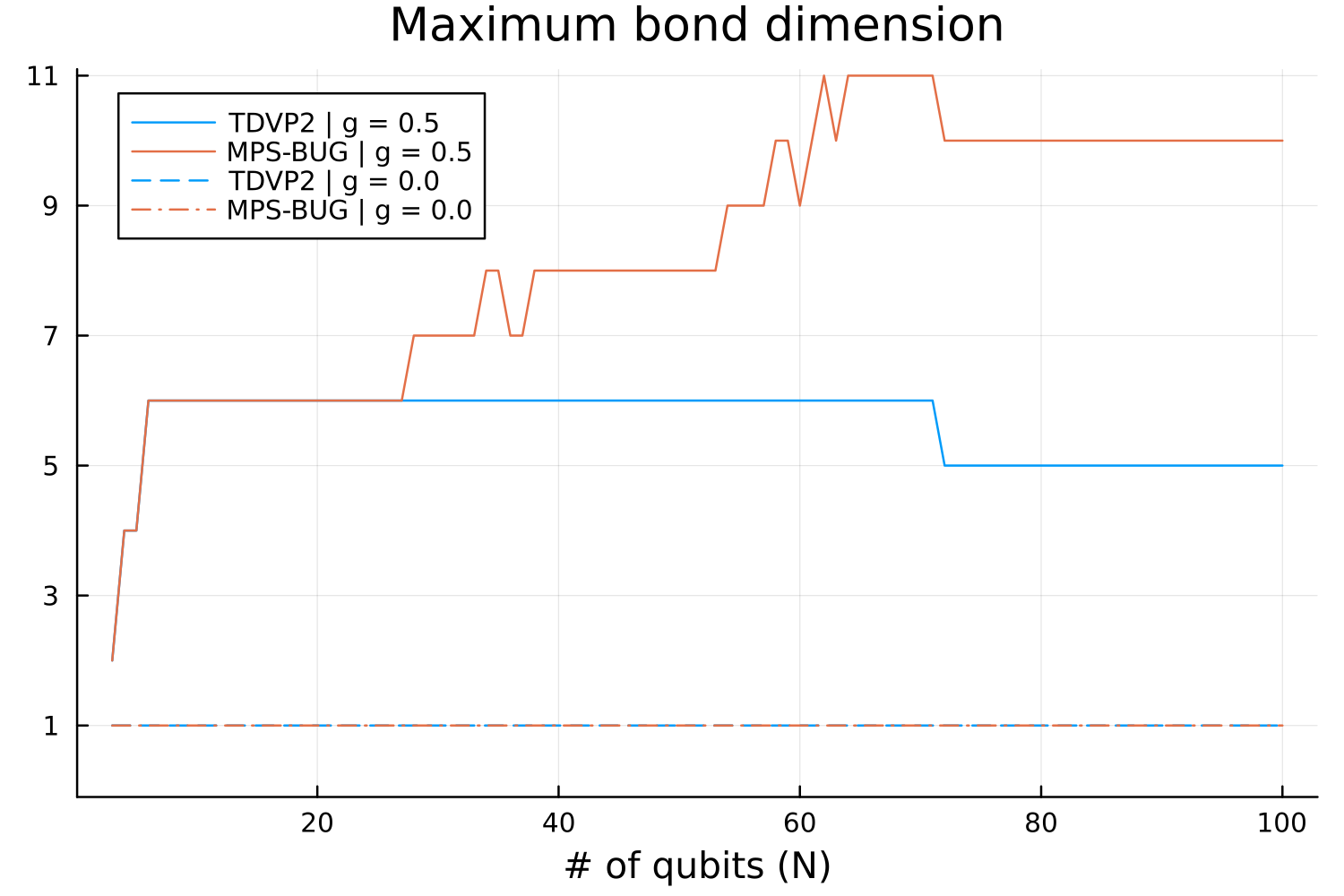}
    \caption{Transverse Ising model with $J=1$ and $g=\{0, 0.5\}$, starting from the ground state, integrated with TDVP-2 and MPS-BUG to time $T=5$. Left: run times as functions of the number of qubits. Right: Maximum bond dimensions.}
    \label{fig:method_convergence} 
\end{figure}
Here, we compare TDVP-2 and MPS-BUG, with and without transverse field terms. Without transverse field ($g=0$), the Hamiltonian is diagonal in the computational basis, and entanglement can only be generated through phase accumulation, and only if the initial state is in a superposition state. However, these simulations start from the ground state, and the solution therefore evolves as a product state. As a result, the bond dimensions are $b=1$ throughout the simulation, explaining the basically linear growth in run time as $N$ increases, for both methods.  
In the case with transverse field terms, $g=0.5$, superpositions are generated for any initial condition, leading to increasing entanglement by the coupling terms, $S^z_j S^z_{j+1}$. In this case, the maximum bond dimensions for TDVP-2 increases initially, but then plateaus at $b=6$ for most $N$. This behavior leads to an almost linear growth in run time for TDVP-2. For MPS-BUG, it's maximum bond dimension is identical to TDVP-2 for $N\leq 26$. For larger numbers of qubits, the bond dimension for MPS-BUG keeps increasing with $N$, leading to longer run times than TDVP-2. The reason for the increase in bond dimension is currently unclear.

Figure~\ref{fig:cutoff_plot_tdvp} (left panel) shows the state error at final time $T=10.0$, as function of the SVD threshold $\epsilon$, for Tucker-BUG, TDVP-2, and MPS-BUG, for a transverse Ising model with $N=10$ qubits. Here each color corresponds to a different time-step. Compared to TDVP-2, twice as many time-steps are needed by both BUG methods to get similar solution errors. This is because TDVP-2 takes two sub-steps within each full time-step.
\begin{figure}
    \centering
    \includegraphics[width=0.49\linewidth]{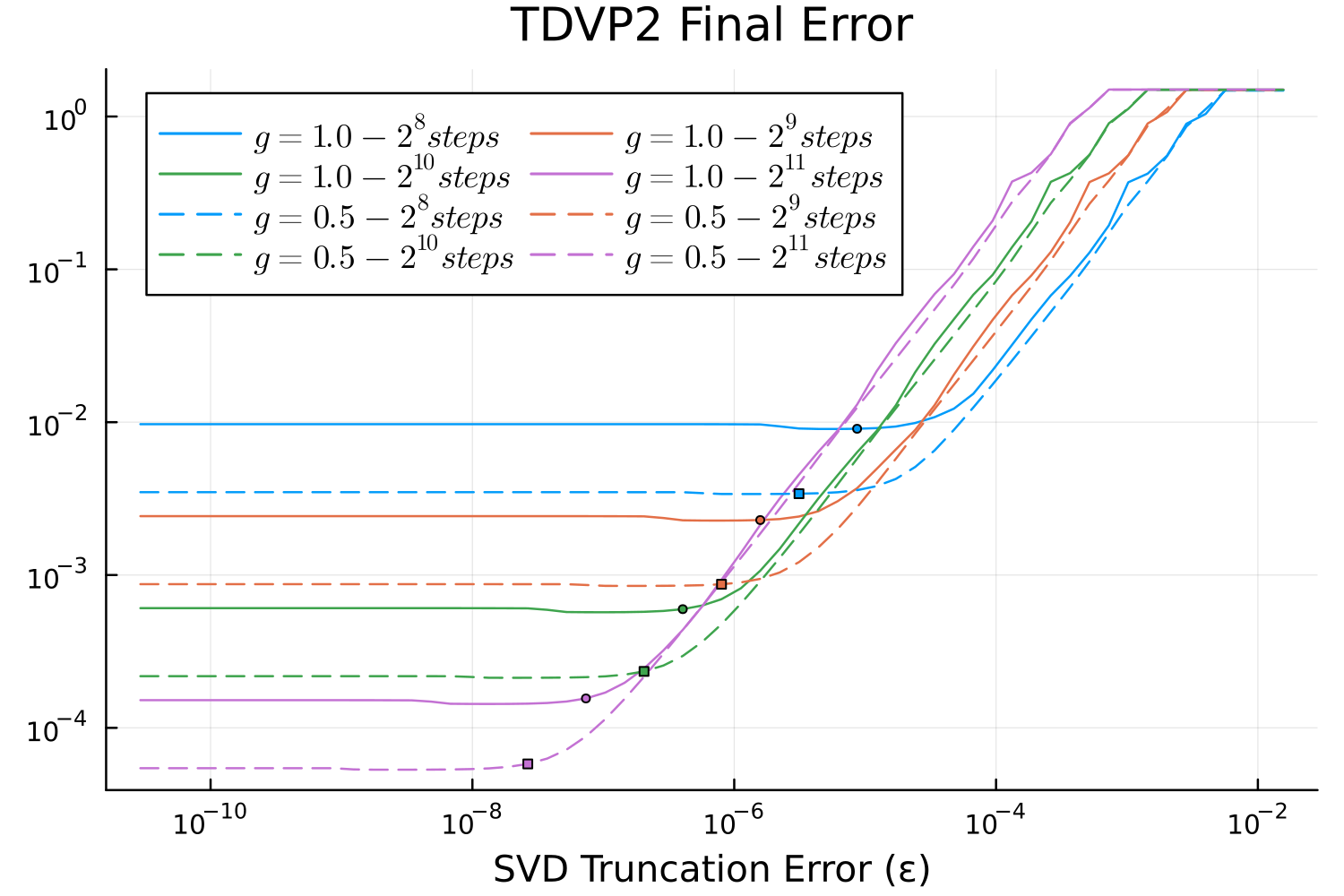}
    \includegraphics[width=0.49\linewidth]{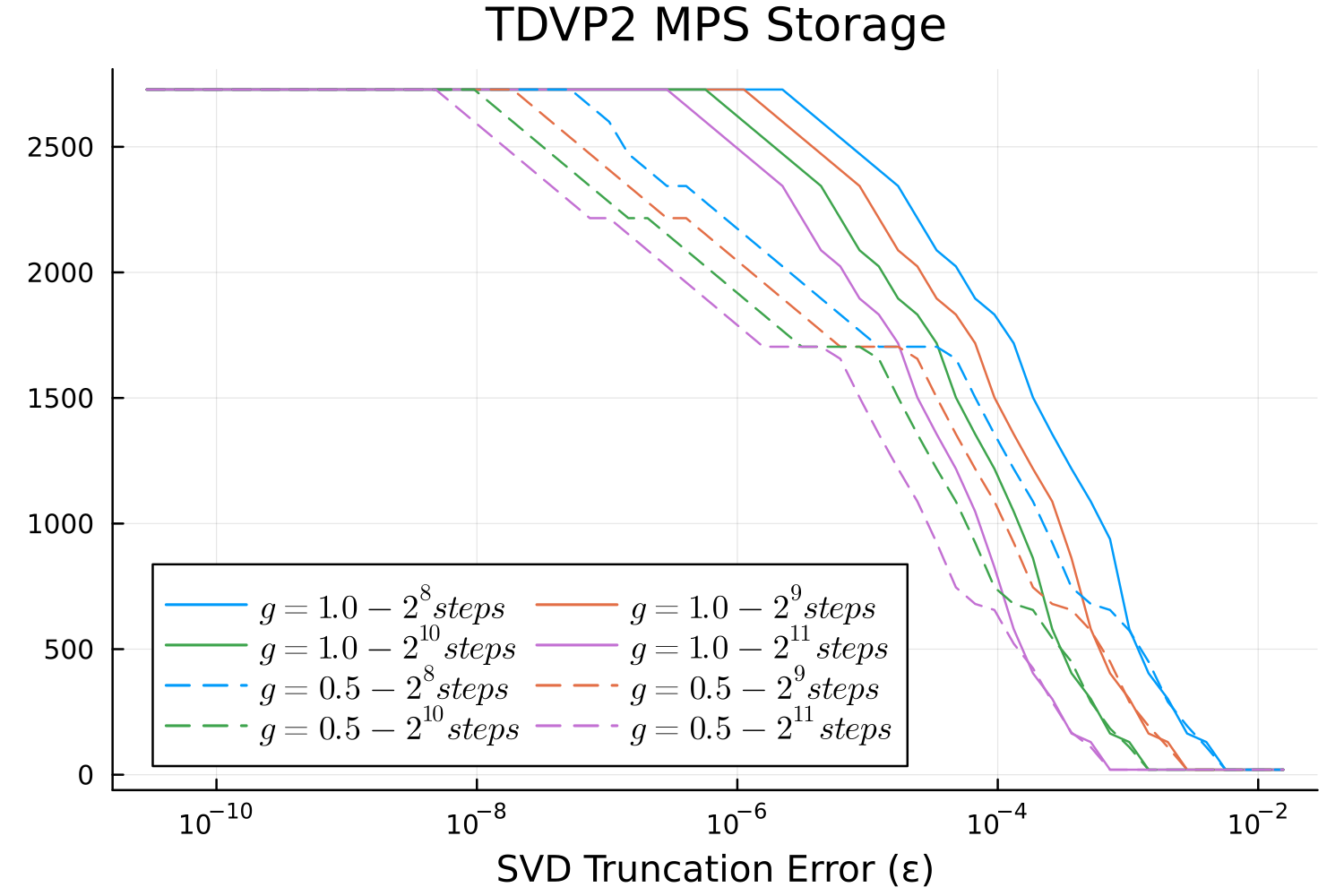}\\
    \vspace{2mm}
    \includegraphics[width=0.49\linewidth]{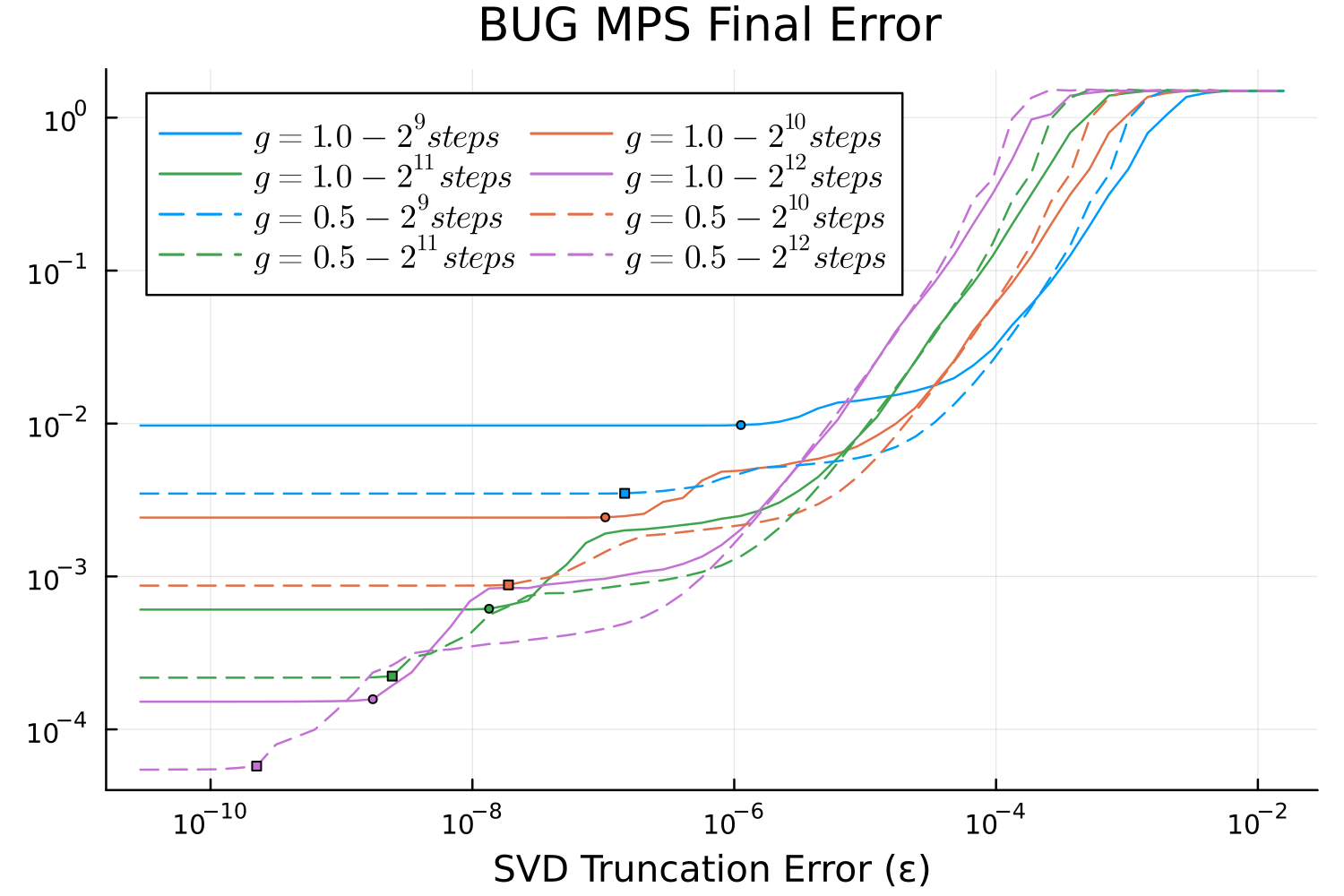}
    \includegraphics[width=0.49\linewidth]{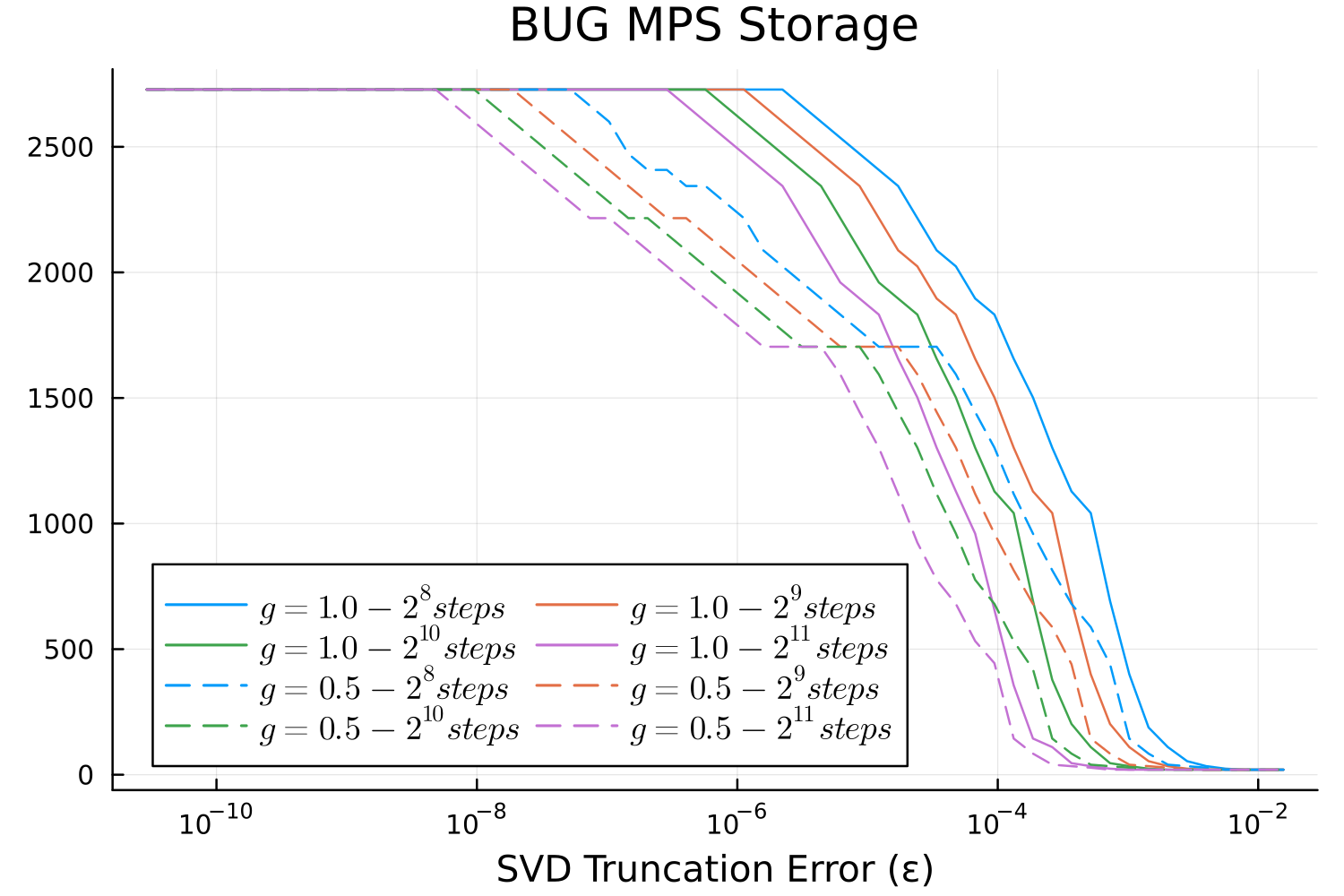}\\
        \vspace{2mm}
    \includegraphics[width=0.49\linewidth]{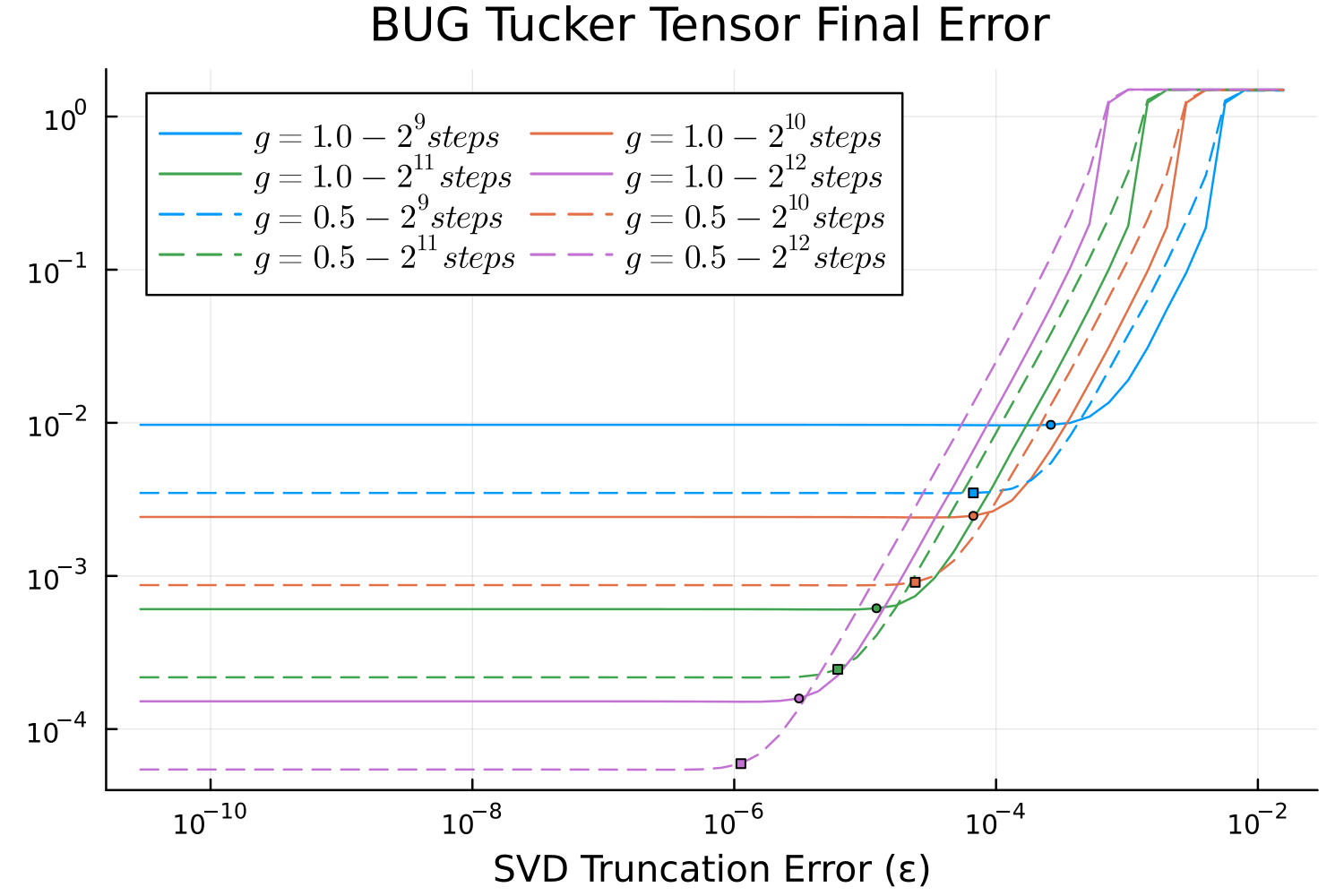}
    \includegraphics[width=0.49\linewidth]{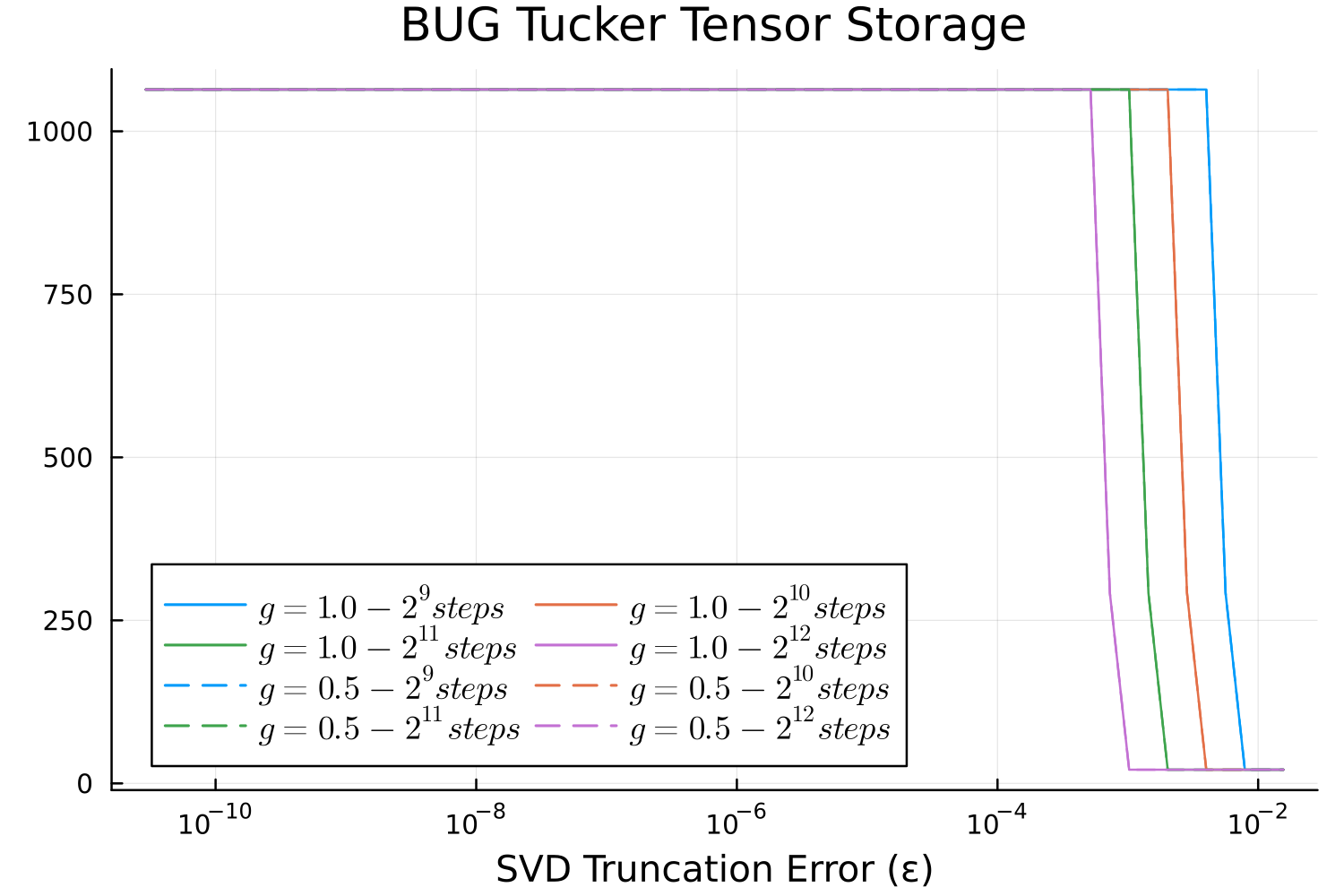}
    \caption{Accuracy and storage requirements as function of the SVD truncation parameter ($\epsilon$) for TDVP-2 (top row), MPS-BUG (middle row), and Tucker-BUG (bottom row). Here, $J=1$ and $g=\{0.5, 1.0\}$ in the Ising model with $N=10$ qubits. Colors indicate different numbers of time-steps, solid lines correspond to $g=1.0$ and dashed lines to $g=0.5$. Left column: State error at final time $T=10.0$. Right column: Maximum number of entries stored for each tensor decomposition.}
    \label{fig:cutoff_plot_tdvp} \label{fig:cutoff_plot_bug} \label{fig:cutoff_plot_tt-bug}
\end{figure}
As the number of time-steps increases, the SVD truncation error accumulates more rapidly. Two regimes can be identified. For sufficiently small $\epsilon$, the accuracy is always dominated by the ${\cal O}(\delta^2)$ time-stepping error. 
For each number of time-steps, the colored squares and circles indicate the largest $\epsilon=\epsilon_{\textrm{trunc}}$ where the solution error is dominated by the ${\cal O}(\delta^2)$ time-stepping error. Hence, using $\epsilon<\epsilon_{\textrm{trunc}}$ only increases the computational effort, without improving the accuracy in the solution. For fixed values of $\epsilon>\epsilon_{\text{trunc}}$, the error can increase when the number of time-steps increases ($\delta$ decreases). From these results we conclude that $\epsilon$ must decrease with $\delta$ to achieve a second order accurate solution. For MPS-BUG, we also note an intermediate range for $\epsilon>\epsilon_\textrm{trunc}$, where the solution is approximately first order accurate.

The maximum number of entries stored in the tensor decompositions during the time-evolution of the Ising model are shown in Figure~\ref{fig:cutoff_plot_tdvp} (right panel). Two trends emerge. For a fixed SVD threshold, increasing the number of time-steps reduces the bond dimensions due to more frequent SVD truncations, leading to reduced storage. However, for a fixed number of time-steps, the storage requirements increase when the SVD threshold is reduced. This highlights the trade-off between accuracy and memory usage inherent to tensor-decomposition methods.

We now analyze the same experiment using the Tucker-BUG integrator, see Figure~\ref{fig:cutoff_plot_bug} (bottom row). As before, when the SVD threshold is sufficiently small, the state error is dominated by the time-stepping error and converges to the same level as for TDVP-2 and MPS-BUG. We also note that the error in the Tucker-BUG method remains small for a larger $\epsilon$ compared to TDVP-2 and MPS-BUG. However, the Tucker-BUG method does not show the same gradual reduction in storage as TDVP-2 and MPS-BUG. This behavior reflects the binary all-or-nothing truncation options for Tucker tensors when applied to two-level (qubit) sub-systems.

Tensor decomposition methods can be used to study the dynamics of observables, e.g., the magnetization at each site,
\begin{equation}\label{eq:magnetization_formula}
    m_j = \langle \psi|M^j|\psi\rangle,\quad j\in[1,N],
\end{equation}
where 
\begin{equation*}
    M^j = 2 I^{\otimes j - 1}\otimes S^z \otimes I^{\otimes N - j}. 
\end{equation*}
In Figure \ref{fig:magnetization-larger-system}, we consider the magnetization in an Ising model with $N=20$ qubits that is evolved to final time $T=15$. The initial data has magnetization $+1$ at site $N$, and $-1$ at all other sites. Panels (a)-(c) show the magnetization using TDVP-2 for different values of the SVD threshold $\epsilon$. By visual inspection, the magnetization does not change between $\epsilon=1.36\cdot 10^{-5}$ (panel (a)) and $\epsilon=2.28\cdot 10^{-4}$ (panel (b)). However, the largest value of $\epsilon= 8.89\cdot 10^{-3}$ (panel (c)) give a very different result, indicating that the SVD truncation lead to a significant error.
\begin{figure}
    \centering
    \begin{subfigure}{0.49\linewidth}
        \centering
        \includegraphics[width=\linewidth]{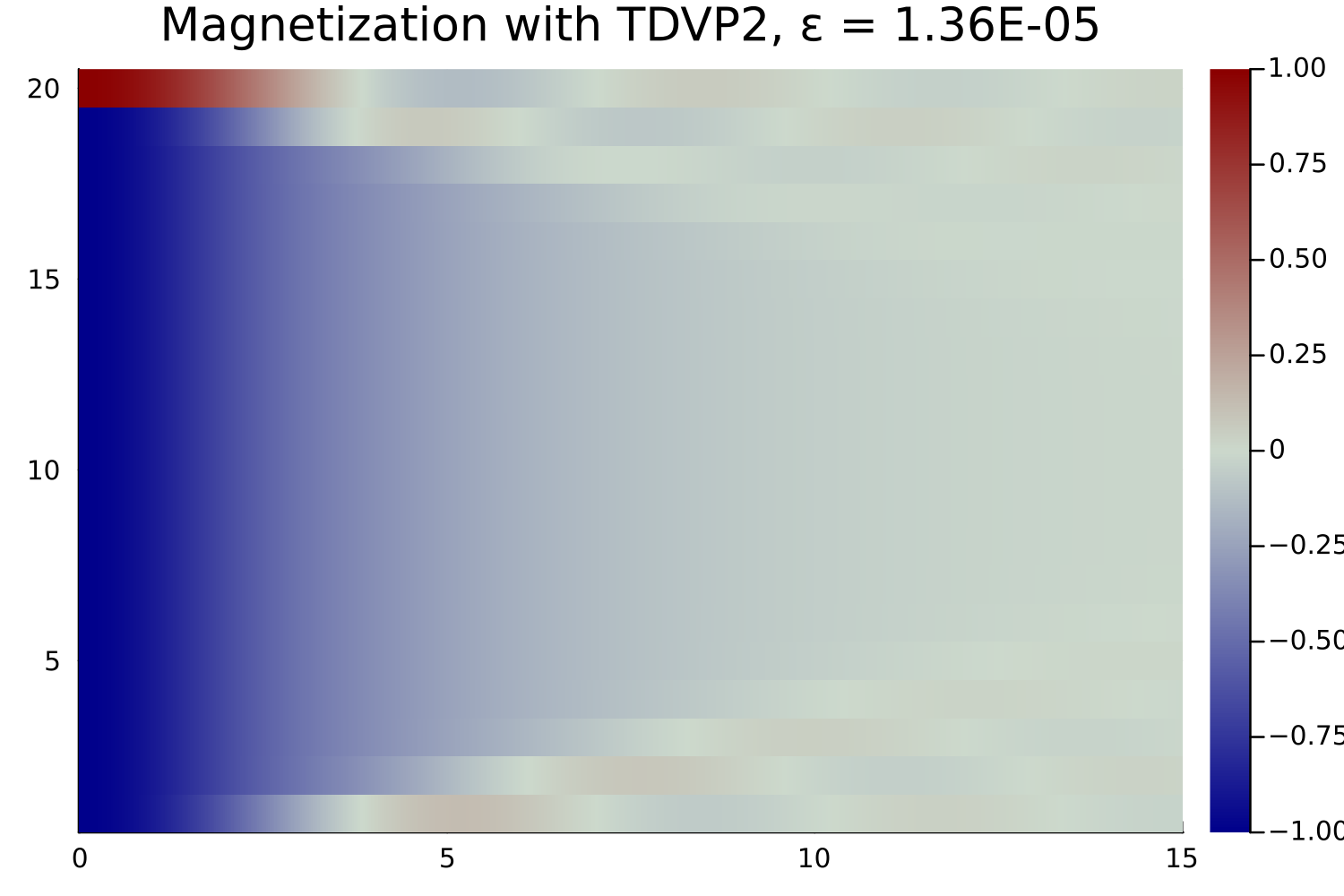}
        \caption{TDVP-2, $\epsilon_1$}
    \end{subfigure}
    \hspace{1mm}
    \begin{subfigure}{0.49\linewidth}
        \centering
        \includegraphics[width=\linewidth]{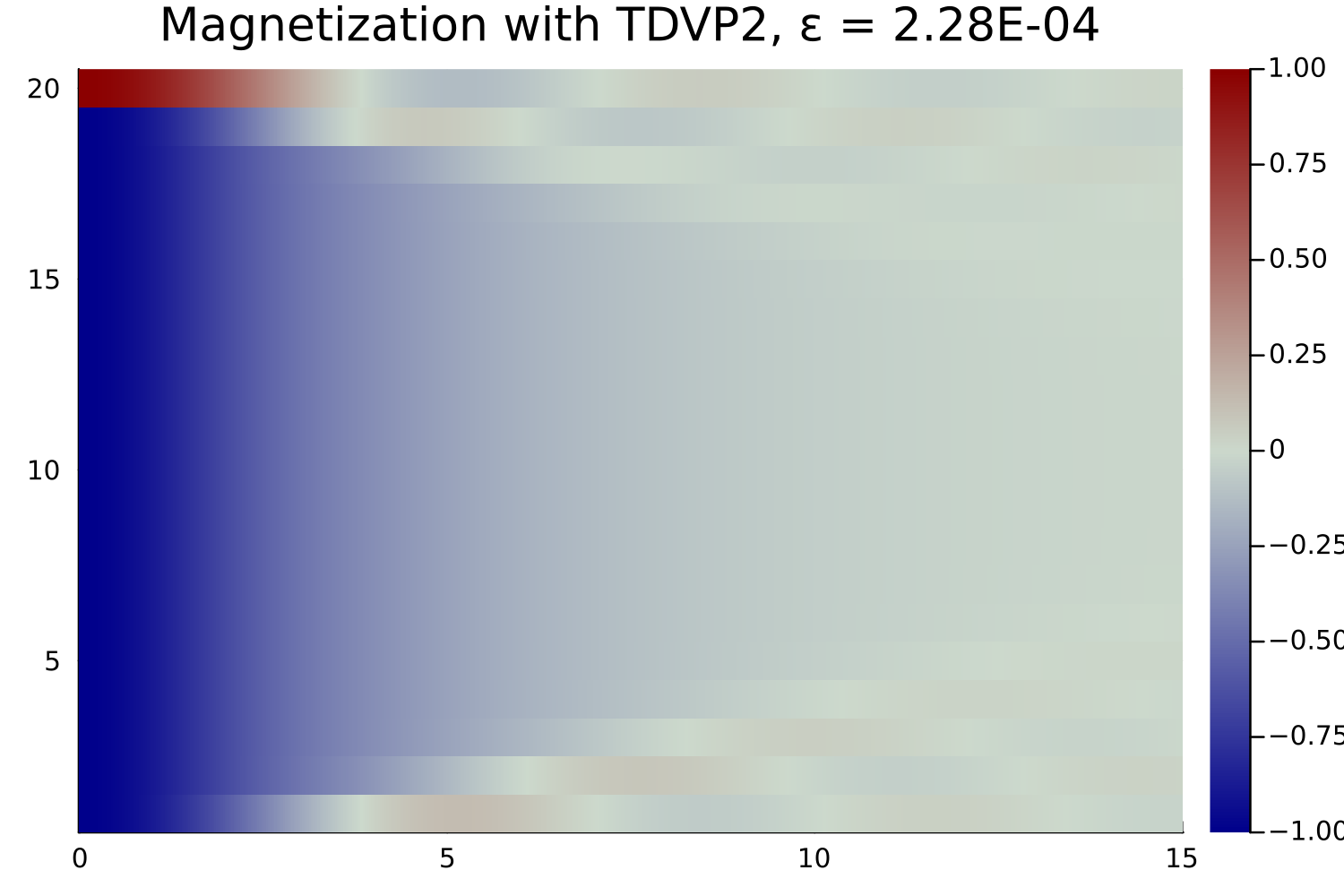}
        \caption{TDVP-2, $\epsilon_2$}
    \end{subfigure}\\
    \vspace{3mm}    
    \begin{subfigure}{0.49\linewidth}
        \centering
        \includegraphics[width=\linewidth]{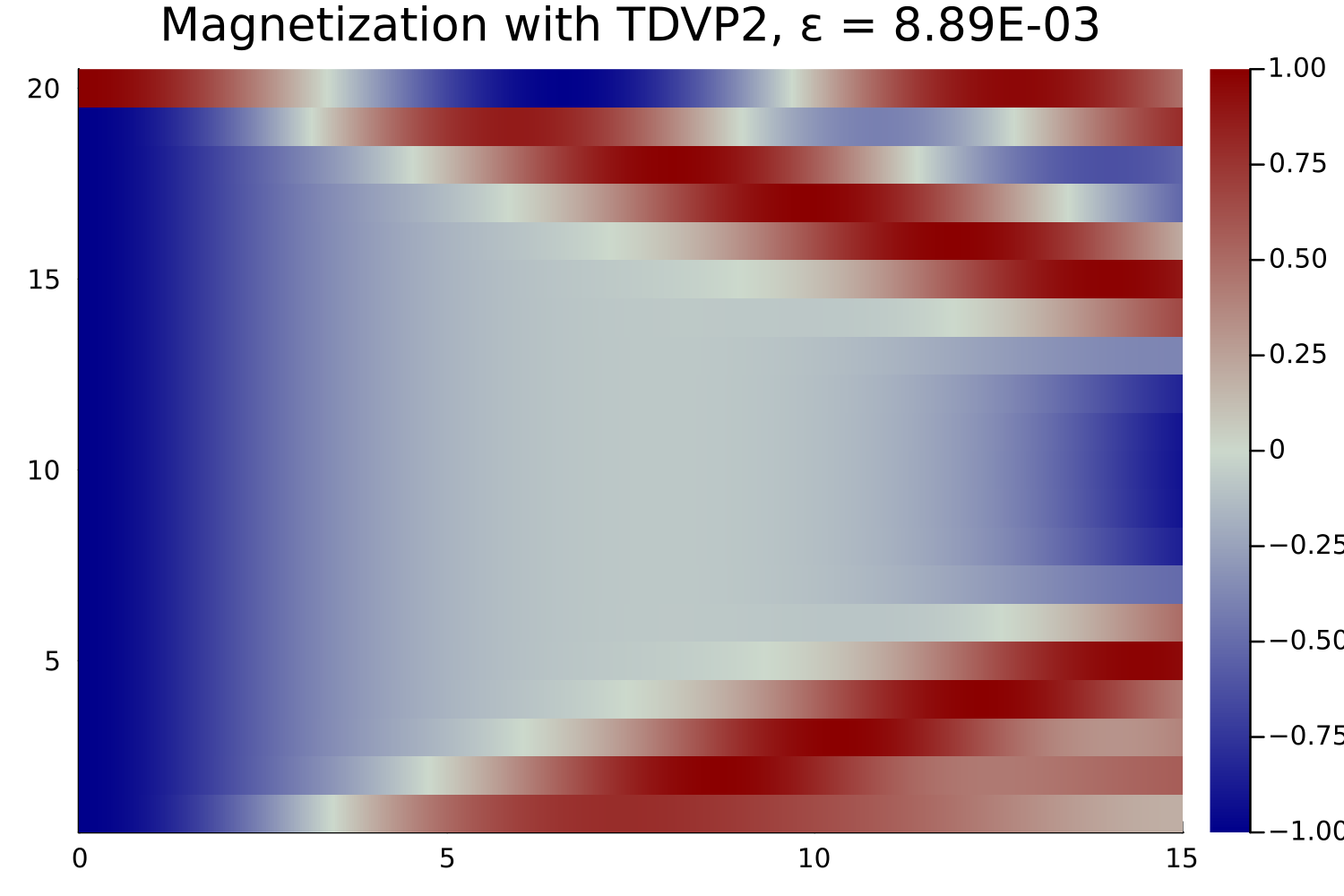}
        \caption{TDVP-2, $\epsilon_3$}
    \end{subfigure}
    \begin{subfigure}{0.49\linewidth}
        \centering
        \includegraphics[width=\linewidth]{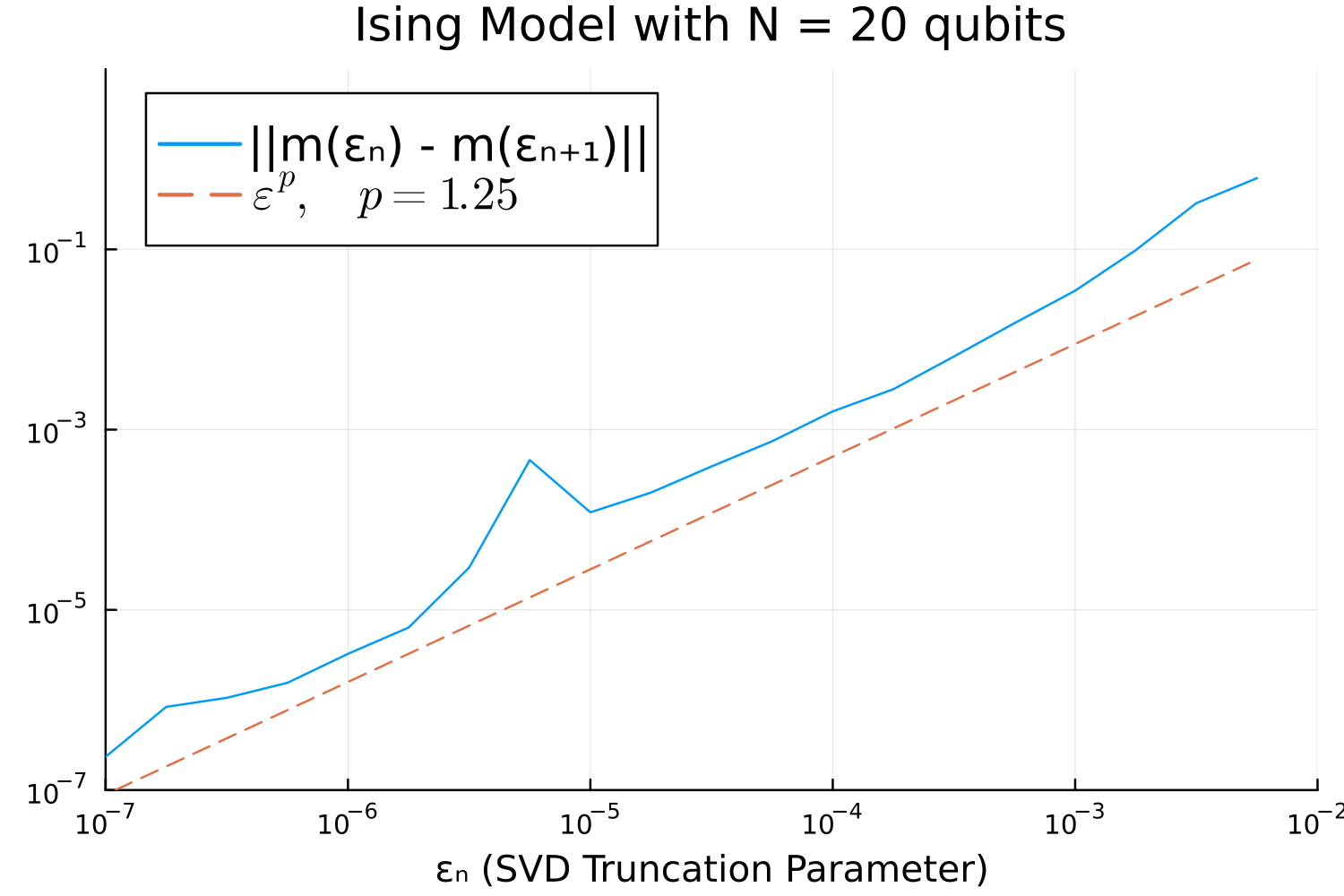}
        \caption{Magnetization difference at $t = 5.0$.}
    \end{subfigure}
    \caption{The magnetization observable in a composite system with $N=20$ qubits evolving under the Ising Hamiltonian with $J = 1.0$ and $g = 0.5$, for different SVD $\epsilon$-thresholds. The initial state has magnetization $m_j=-1$ for $j\in[1,19]$ and $m_{20}=+1$, and the system is evolved to final time T = 15.0 using $256$ time-steps. In the bottom-right panel, we evaluate the magnetization vector at time $t=5.0$ for $\epsilon_k = 10^{-3}\cdot 10^{-k/4}$ for $k\in[1,20]$, and report the norm of the difference $\|\vec{m}_k - \vec{m}_{k-1}\|$ as function of $\epsilon_k$.}
    \label{fig:magnetization-larger-system}
\end{figure}
In this case, the state has $N_{tot}\approx 10^6$ elements, making it intractable to calculate an exact solution via matrix exponentiation on the available hardware. To illustrate how the $\epsilon$-threshold in the SVD truncation influences the magnetization, we evaluate the magnetization vector $\vec{m} = [m_1,\ldots m_N]$ at time $t=5$ as function of $\epsilon$. Starting with $\epsilon_0=10^{-3}$ we define $\epsilon_{k} = \epsilon_{k-1} \cdot 10^{-1/4}$, for $k\in[1,20]$. This gives $\epsilon_k  = \epsilon_0 \cdot 10^{-k/4}$ and $\epsilon_{20} = 10^{-8}$. We denote the final magnetization, calculated with $\epsilon=\epsilon_k$, by $\vec{m}_k$ and then study the norm of differences in magnetization, $\Delta_k := \| \vec{m}_{k} - \vec{m}_{k-1} \|$, for $k=1,2,\ldots,20$, see Figure~\ref{fig:magnetization-larger-system}(d). The slope is $\approx 1.25$ indicating that the error in magnetization satisfies $\vec{m}_k - \vec{m}_{\text{exact}} = {\cal O}(\epsilon_k^p)$, with $p=1.25$.

\subsection{Time-Dependent Hamiltonians}

Time-dependent Hamiltonians are common in quantum control applications, where they are used to model quantum dynamics due to external control pulses applied to the system. In the following we consider the Hamiltonian model for a linear chain of $N$ coupled superconducting transmon qubits stated in~\eqref{eq:td-hamiltonian}-\eqref{eq:control-hamiltonian}. The dimension of the Hilbert (vector) space of each sub-system is $d_k=2$, resulting in a state vector $|\psi\rangle \in \mathbb{C}^{2^N}$. The MPO representation of this Hamiltonian has bond dimension equal to $4$, for all sites. In the examples below, we take the frequency of rotation, $\omega^d$, to be the arithmetic mean of the transition frequencies $\{\omega_k\}_{k=1}^N$. The time-dependence of the control functions $p^k(t)$ and $q^k(t)$ is assumed to be a sum of harmonic carrier waves with frequency $\Omega_k$, where the envelope function is parameterized by B-spline wavelets. 

As a test case, we consider the state-to-state problem where the state of each qubit is transformed from the $N$-qubit ground state into the superposition state,
\begin{align}\label{eq_state-2-state}
    |0\rangle^{\otimes N}
    \;\longrightarrow\;
    \left(\tfrac{|0\rangle + |1\rangle}{\sqrt{2}}\right)^{\otimes N}.
\end{align}
Due to the inherent coupling between the qubits, it would be challenging to analytically determine the control functions for the state-to-state transformation. Instead, we numerically determine the control pulses in the control Hamiltonian \eqref{eq:control-hamiltonian} using the Quandary code~\cite{QuandaryUser}; for details of the numerical approach, see~\cite{GunPetDub-2021,BSplines}. In this approach, Quandary simultaneously optimizes for the control pulses in all qubits. For example, in the case of $N=6$, the optimized control pulse for qubit \#1 is shown in Figure~\ref{fig:pulse_6}.
\begin{figure}
    \centering
    \includegraphics[width=0.49\linewidth]{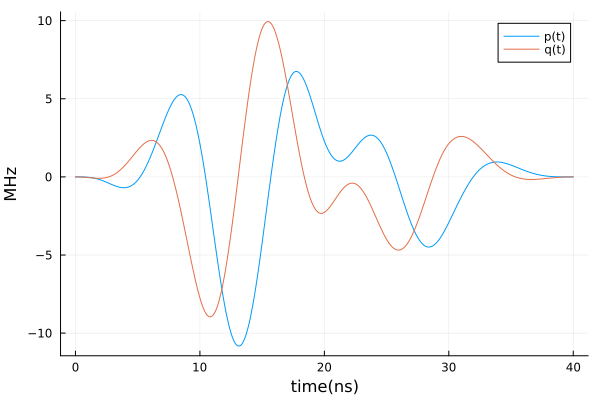}
    \caption{Real and imaginary parts of the control pulse applied to qubit \#1, in the case of $N = 6$ qubits in the system.}
    \label{fig:pulse_6}
\end{figure}

In the following, we test the numerical accuracy of the TDVP-2 and Tucker-BUG algorithms by comparing the state and infidelity with those obtained with Quandary. Due to the time dependence in the control Hamiltonian, the evolution of the quantum system cannot be evaluated analytically. Instead, a reference solution is obtained using Quandary.

\subsection{Coupled systems with optimized control pulses}

We consider the case of short-range interaction in the Hamiltonian and set the coupling strength between qubits $k$ and $\ell$ to be
\[
\frac{J_{k\ell}}{2\pi} = \begin{cases}
    5.0\, \text{MHz},& \ell = k+1,\\
    0.0, & \text{otherwise}.
\end{cases}
\]
Here, we choose the transition frequencies to be $\omega_k/2\pi = 4.64 + 0.06 (k-1)$ GHz, for $k\in[1,N]$. Because of the coupling between qubits, we use three carrier frequencies for each qubit (two for the first and last qubits in the chain). The carrier frequencies are $\Omega_{k} = \{\omega_{k-1}, \omega_k, \omega_{k+1}\} - \omega^d$.

To establish a reference solution, we evolve this system with Quandary to final time $T = 40.0$ and used Richardson extrapolation to estimate the error in the final state vector. For $N=6$ qubits, with 4152 time-steps, the norm of the error was $\approx 2.97\cdot 10^{-5}$. For $N=10$ qubits, with 10,480 time-steps, the norm of the error became $\approx 7.82\cdot 10^{-5}$.

Figure \ref{fig:qft_err-storage-combo} (left column) shows the difference between the states computed with tensor decomposition methods (TDVP-2, Tucker-BUG) and Quandary, as function of the $\epsilon$-threshold in the truncated SVD. We note that the infidelity computed with the TDVP-2 and Tucker-BUG methods is at least one order of magnitude smaller than the norm of the state-differences. Also note that the fidelity obtained by both tensor decompositions agree with Quandary when $\epsilon\leq 10^{-5}$ ($N=6$) and $\epsilon\leq 5\cdot 10^{-6}$ ($N=10$), respectively. Figure \ref{fig:qft_err-storage-combo} (right column) shows the total number of entries stored in the tensor decompositions compared to the vector storage in Quandary. In this figure we note the all-or-nothing truncation property of the Tucker-BUG algorithm as the bond dimension in the factor matrices drops from 2 to 1 at $\epsilon \approx 10^{-4}$. Compared to Quandary, TDVP-2 uses less storage when $\epsilon> 5.5\cdot 10^{-5}$ ($N=6$) and when $\epsilon > 10^{-9}$ ($N=10$), respectively. For the $N=10$ case, we conclude that TDVP-2 needs less storage than Quandary to calculate both the final state and the fidelity with the same accuracy.
\begin{figure}
    \centering
    \includegraphics[width=0.49\linewidth]{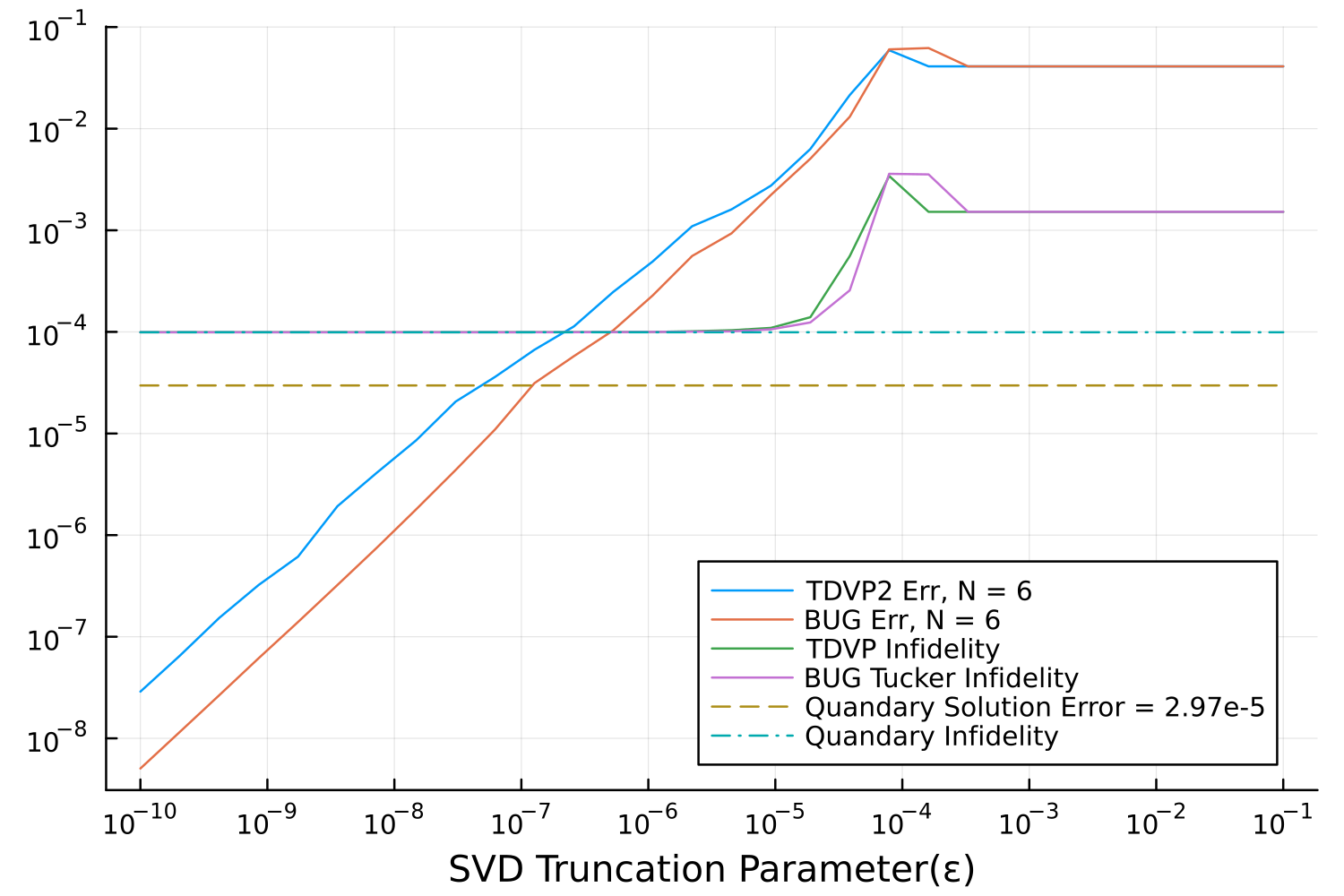}
    \includegraphics[width=0.49\linewidth]{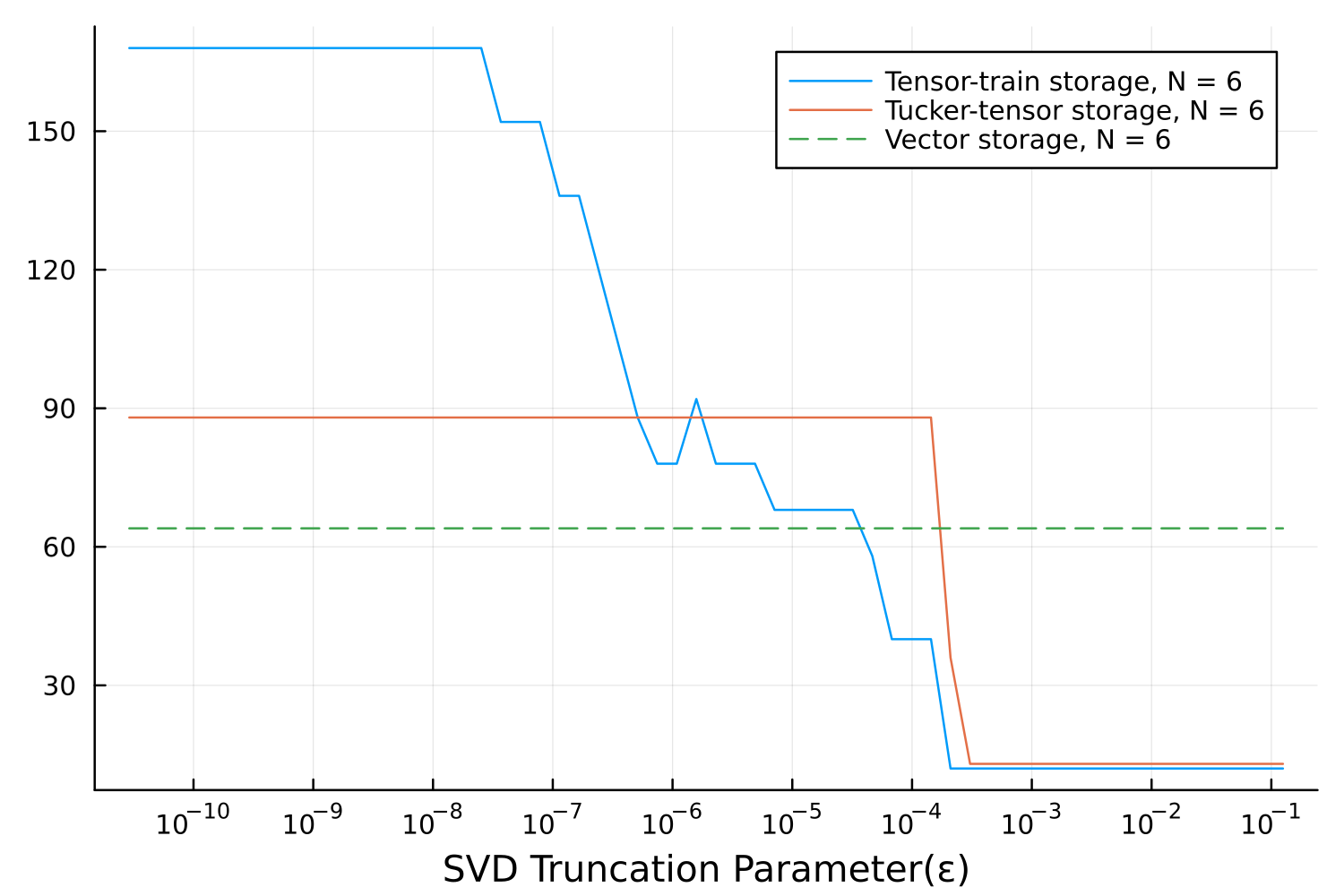}\\
    \includegraphics[width=0.49\linewidth]{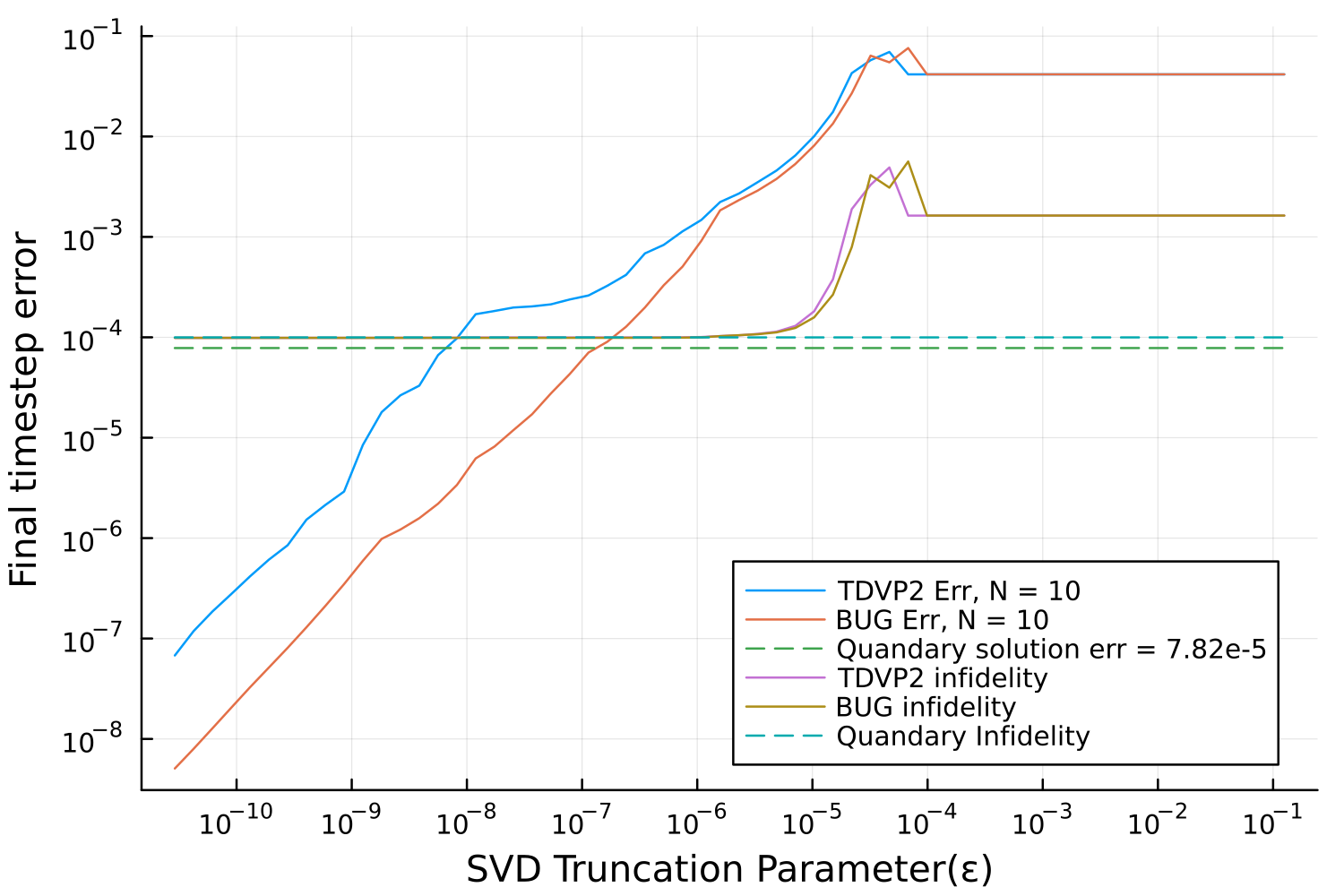}
    \includegraphics[width=0.49\linewidth]{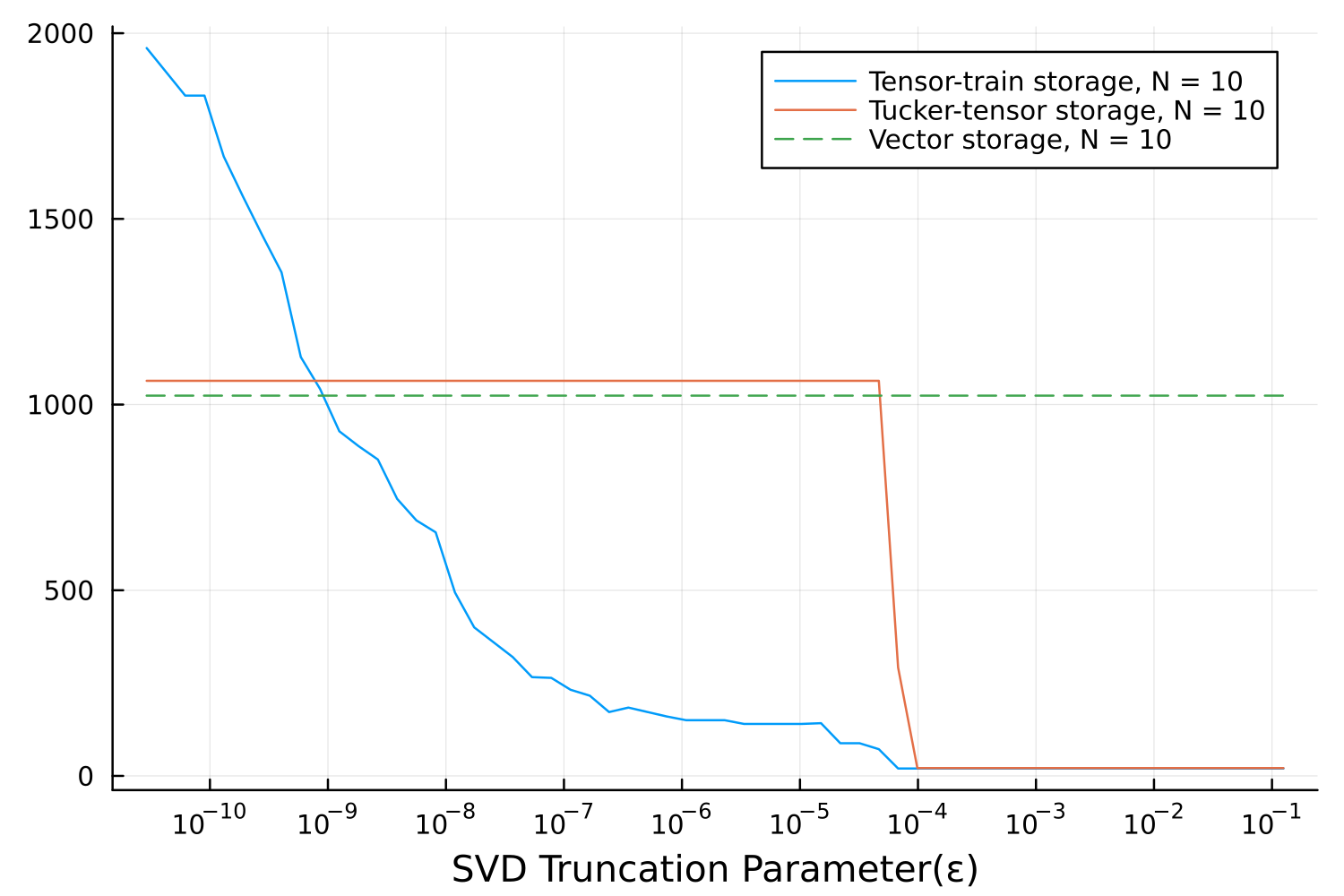}
    \caption{Solution accuracy (left column) and storage requirements (right column) at the final time $T=40$ for systems with $N = 6$ qubits (top row) and $N=10$ qubits (bottom row). Left column shows the norm of the differences between the states from the tensor-decomposition methods and Quandary, as function of the $\epsilon$-threshold in the truncated SVD. Also shown are the infidelities from TDVP-2, Tucker-BUG, and Quandary. Right column shows the total number of entries stored in the MPS and the Tucker-tensor as function of $\epsilon$. The dashed green-line indicates the storage needed for one state vector in Quandary.}
    \label{fig:qft_err-storage-combo}
\end{figure}

\subsection{Larger systems of qubits with optimized control pulses}
Consider a system with $N$ qubits in a linear chain where the transition frequencies oscillate with period 8:
\begin{align*}
    r(k) = (k-1) \bmod 4,\quad
    \omega_k =
    \begin{cases}
        5.18 - 0.06\, r(k),
        & k=1,2,3,4,\,9,10,11,12, \ldots \\
        5.18 + 0.06\, r(k) - 0.15,
        & k=5,6,7,8,\, 13, 14, 15, 16, \ldots
    \end{cases}
\end{align*}
Because it becomes difficult to numerically optimize the control pulses as the number of qubits becomes larger, we here consider the case with $J_{k\ell} =0$. This allows the control pulses for the state-to-state transformation \eqref{eq_state-2-state} to be optimized independently for each qubit. As a result, only 1 carrier frequency is needed per qubit, $\Omega_k = \omega_k - \omega^d$. 

Figure \ref{fig:runtime_scaling} shows a run time comparison between TDVP-2 and Quandary for an increasing number of qubits.
\begin{figure}
    \centering
    \includegraphics[width=0.49\linewidth]{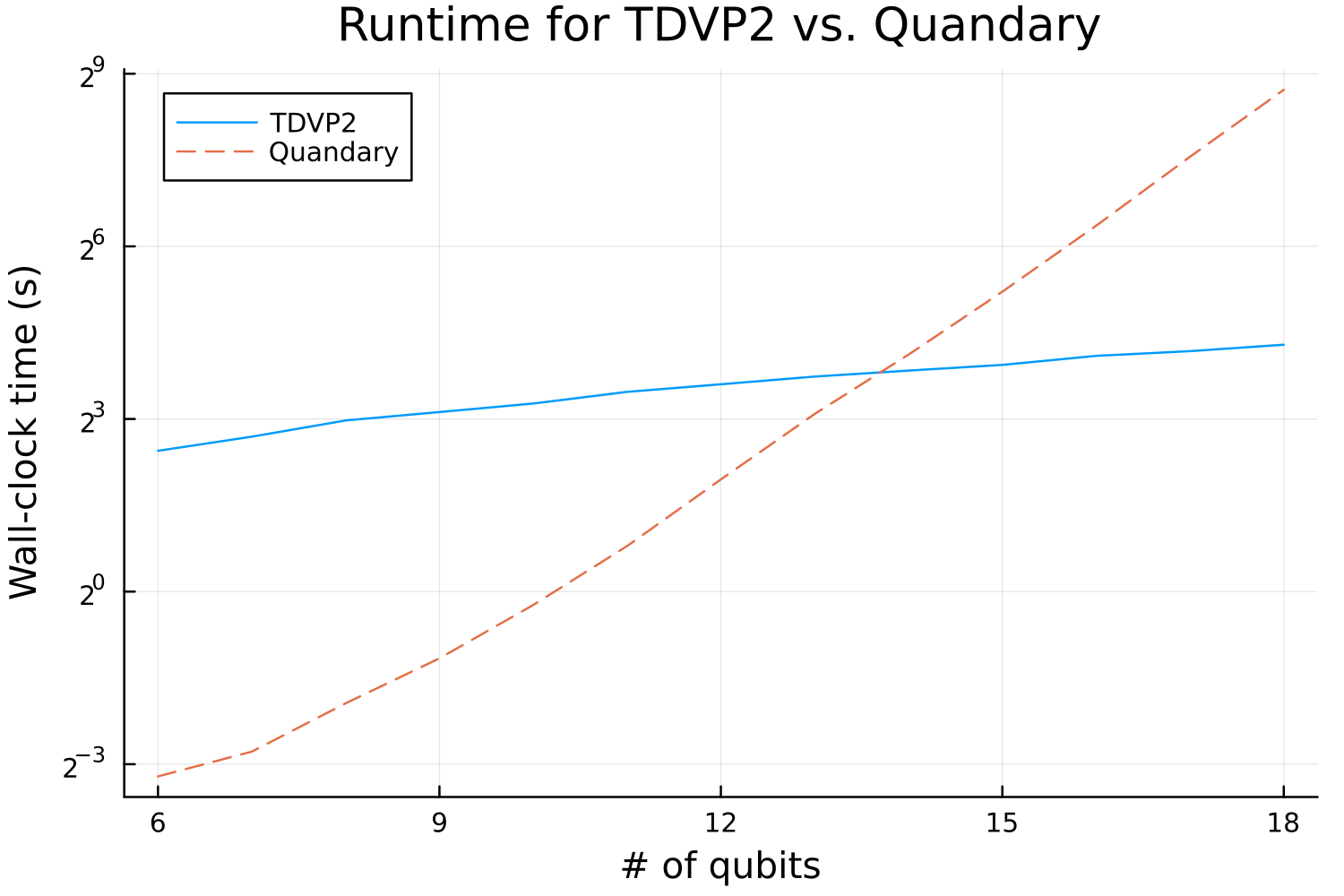}
    \includegraphics[width=0.49\linewidth]{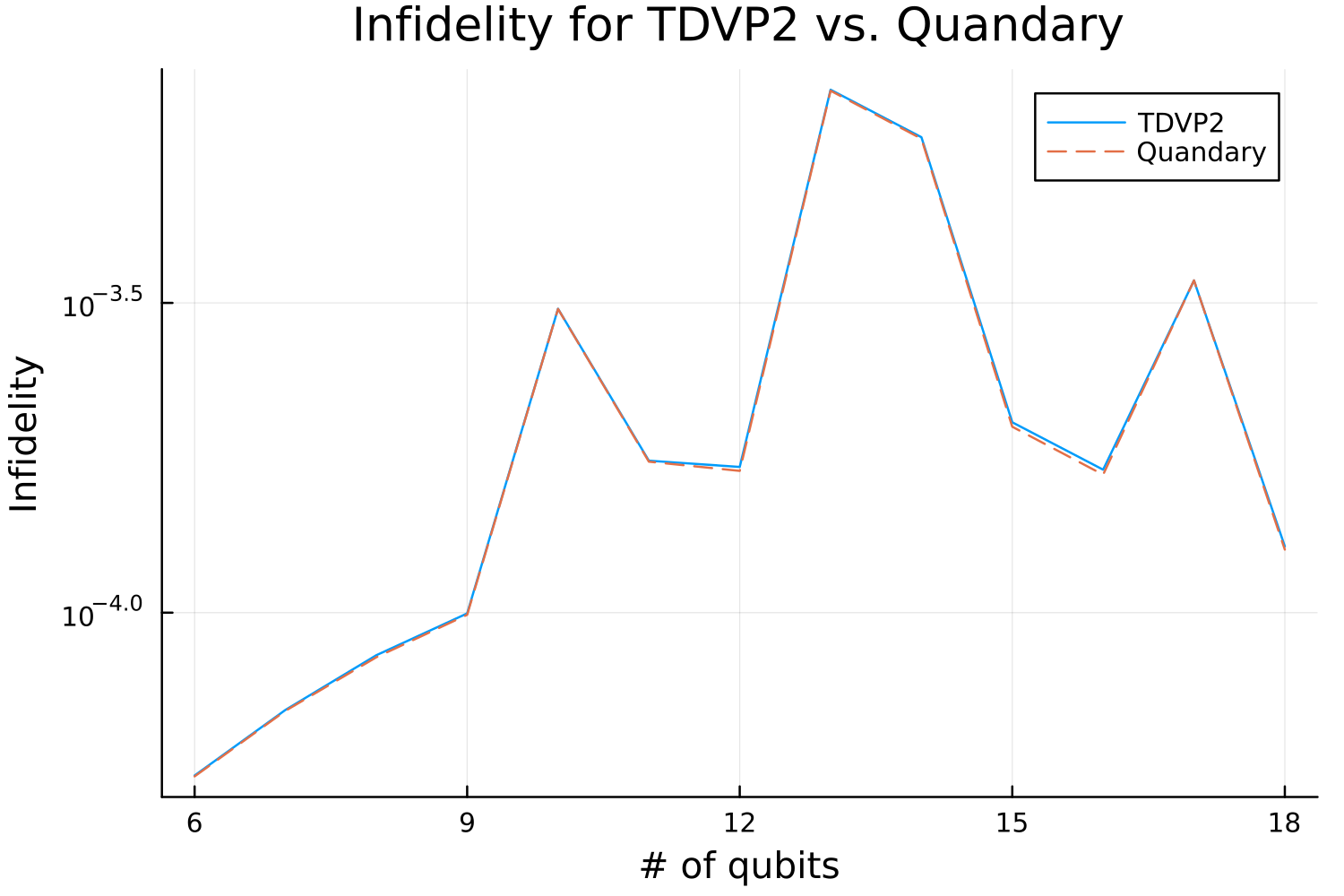}
    \caption{Comparing Quandary and TDVP-2 on a state-to-state transformation problem for different numbers of qubits with fixed SVD threshold, $\epsilon = 10^{-6}$. Left: run time. Right: infidelity.}
    \label{fig:runtime_scaling}
\end{figure}
To obtain a fair comparison with TDVP-2, which is implemented as a serial code, we also execute Quandary on a single core. The Tucker-BUG integrator was omitted from these test because the algorithm becomes very slow for $N\geq 10$ qubits. 

As is shown in Figure \ref{fig:runtime_scaling}, Quandary is faster than TDVP-2 for systems with 6-13 qubits. Past 13 qubits, TDVP-2 becomes faster than Quandary. In Quandary, the exponential increase of the state-vector leads to an exponential increase in run time. With TDVP-2, the bond dimension can be kept small and independent of $N$ because there is no coupling between the qubits ($J_{k\ell}=0$). This results in run times that only increase linearly in the number of qubits. In Figure \ref{fig:runtime_scaling} (right), we compare the infidelities from TDVP-2 and Quandary. We conclude that both techniques result in almost identical results.

\section{Conclusions}\label{sec_conclusions}

In this paper we have derived and numerically evaluated tensor train and Tucker tensor decomposition methods for solving Schr\"odinger's equation in the time domain. For time-integration of tensor trains, we reviewed the classical TDVP and TDVP-2 algorithms, and also outlined the MPS-BUG  algorithm. For Tucker decompositions, we sketched the Tucker-BUG algorithm. In numerical experiments we have evaluated the accuracy, memory usage and run times of the rank-adaptive TDVP-2, MPS-BUG and Tucker-BUG methods. The rank-adaptivity in all the methods is based the $\epsilon$-truncated singular value decomposition (SVD). Solution accuracy and run times have been compared with the conventional matrix-vector formulation implemented in the Quandary~\cite{GunPetDub-2021, QuandaryUser} software.

Numerical experiments based on composite quantum systems consisting of two-level sub-systems (qubits) were presented for two Hamiltonian models: the time-independent transverse Ising model, and a model of transmon qubits subject to time-dependent control pulses. For the two-level cases considered here, Tucker tensors exhibit an inherently binary behavior in the $\epsilon$-truncated SVD, where the factor matrices either have full size ($2 \times 2$), or the truncated size ($2 \times 1$). As a result, each sub-systen is either treated as fully entangled, or completely decoupled, from the rest of the system. Tucker tensors are therefore of limited use for two-level systems, but may be more effective for quantum systems with higher local dimensions, where more gradual truncation is possible.

For the transverse Ising model with short-range interactions, we have demonstrated that quantum simulations with tensor trains can be tractable for systems with 100 qubits (or more), using a laptop computer. This is possible because of the limited growth in entanglement during the evolution, leading to run times that grow as low-order polynomials, instead of exponentially with conventional matrix-vector formulations. For modest numbers of sub-systems, where the accuracy can be evaluated via exact solutions based on matrix exponentiation, we have analyzed the effect of the $\epsilon$-threshold in the truncated SVD, and how it needs to be related to the time-step. We also studied how the $\epsilon$-truncation influences the magnetization observable in a 20-qubit system.

When the state is entangled, the MPS-BUG method is observed to be slower than TDVP-2. A part of this can be attributed to the SVD factorizations during the truncation of the inflated bond dimensions at the end of each time-step. However, we have also observed that the bond dimension in MPS-BUG sometimes becomes larger than in TDVP-2, especially when the number of coupled sub-systems is large. The reason for this behavior is currently poorly understood.

The performance of the tensor decomposition methods for time-dependent Hamiltonian models was evaluated on a state-to-state transformation problem. Here, we considered a quantum system of $N$ coupled transmon qubits subject to time-dependent control pulses. Reference solutions were generated with the Quandary code and used to estimate the accuracy of the tensor decomposition methods. We studied how the $\epsilon$-threshold in the truncated SVD affects the accuracy and storage requirements in the TDVP-2 and Tucker-BUG algorithms. For an $\epsilon$ that gives the same accuracy as Quandary, the TDVP-2 algorithm uses less storage than Quandary, if $N$ is sufficiently large. For the same $\epsilon$, the Tucker-BUG algorithm gave similar accuracy as TDVP-2, but was slower and used significantly more storage. Finally, we compared the run time performance of TDVP-2 and Quandary as the number of qubits increases. In this case, we observed that TDVP-2 outperforms Quandary when the number of qubits exceeds $N\approx 13$, without degrading the fidelity of the state-to-state transformation.

One natural extension of the current work is to consider qibits that are connected in a two-dimensional lattice. Another interesting extension is to apply tensor train methods to model quantum decoherence by solving Schr\"odingers equation with a non-Hermitian Hamiltonian, e.g., within the Monte-Carlo wavefunction method.

\section*{Acknowledgments} The authors gratefully acknowledge financial support from the Office of Science, Advanced Scientific Computing Research (ASCR) within the U.S.~Department Of Energy, award \#~SCW1886. C.~Hodges-Heilmann was also supported by the U.S.~Department of Energy, Office of Science, ASCR, under award \#~DE-SC0025481, and by the National Science Foundation under grant \#~DMS-2436318. The authors thank Dr.~Dominik Sultz and Dr.~Jonas Kusch for helpful discussions on the BUG algorithm for tensor trains.

This work was performed under the auspices of the U.S. Department of Energy by Lawrence Livermore National Laboratory under Contract DE-AC52-07NA27344. This is contribution LLNL-JRNL-2016933.

\appendix

\section{The singular-value decomposition} \label{app:svd}
Consider a matrix $M \in \mathbb{C}^{m \times n}$. The (compact) Singular Value Decomposition (SVD) of $M$ is given by
\begin{align}
    M = U\Sigma V^{\dag},\quad
    U^{\dag}U = I_{r}, \quad V^{\dag}V = I_r,\quad 
    \Sigma = \textrm{diag}[\sigma_1, \sigma_2, ..., \sigma_{r}],\quad r\leq \min(m,n),\label{eq_SVD}
\end{align} 
where $U\in\mathbb{C}^{m\times r}$, $V\in\mathbb{C}^{n\times r}$, $\Sigma\in\mathbb{R}^{r\times r}$, and $r$ is the rank of $M$. The elements of $\Sigma$ are real and positive and are called the singular values of $M$. In the following we assume them to be sorted in descending order, $\sigma_1 \geq \sigma_2 \geq \cdots \geq \sigma_{r}$. 

Several algorithms for constructing low-rank approximate representations of tensors are based on low-rank approximations of the SVD. Denote by $\widetilde{M}$ an approximation of the matrix $M$, formed by setting the smallest singular values to zero,
\begin{equation}\label{eq:zero_singular_values}
    \widetilde{\Sigma} = \textrm{diag}[\sigma_1, ..., \sigma_k, 0, ..., 0].
\end{equation}
Because the matrices $U$ and $V$ have orthonormal columns, the Frobenius norm difference can be calculated exactly, 
\begin{align}
    \|M - \widetilde{M}\|_F &= \|U\Sigma V^{\dag} - U\widetilde{\Sigma}V^{\dag}\|_F =
    \|\Sigma - \widetilde{\Sigma}\|_F = \sqrt{\sum_{i = k + 1}^{r}\sigma_i^2}\leq \epsilon.\label{eq:eq_epsilon}
\end{align}
Here, $k\geq 1$ is the largest index such that the inequality holds for a given $\epsilon\geq 0$.
The matrix $\widetilde{M}$ has $k$ non-zero singular values and therefore has rank $k$. One can show that $\widetilde{M}$ is the rank-$k$ matrix that is the closest to $M$ in Froebenius norm, thus being the optimal approximation of $M$. 

\subsection{Order-4 Tensor SVD}
In TDVP-2 a SVD is performed on an order-4 tensor to decompose it into two order-3 tensors and move the orthogonality center to the right. This is done by reshaping the tensor into a matrix, performing a SVD on that matrix, followed by inverse reshape operations. 
For example, let $T \in \mathbb{C}^{b_i \times d_i \times d_{i + 1} \times b_{i + 2}}$ represent the merged core for sites $i+1$ and $i+2$. Reshape $T$ into a matrix,
\begin{equation*}
    \hat{T} := \text{reshape}(T) \in \mathbb{C}^{b_i d_i \times d_{i + 1}b_{n + 2}}.
\end{equation*}
Perform a SVD on $\hat{T}$,
\begin{align}
    \hat{U}\hat{S}\hat{V}^{\dag} := \hat{T},\quad
    \hat{U} \in \mathbb{C}^{b_i d_i \times s_i}, \, \hat{S}\in\mathbb{C}^{s_i \times s_i}, \, \hat{V}^{\dag} \in \mathbb{C}^{s_i \times d_{i + 1}b_{i + 2}}.\label{eq:order_4_svd}
\end{align}
Here we have $s_i = \min\{b_i d_i, d_{i + 1}b_{i + 2}\}$, and we set the intermediate bond dimension to $b_{i + 1} = s_i$. The inverse reshape operations is then performed on the matrix $\hat{U}$ and on the matrix $\hat{R} := \hat{S}\hat{V}^{\dag}$, where $\hat{R}\in \mathbb{C}^{s_i \times d_{i + 1}b_{i + 2}}$,
\begin{align}
    A_{[i+1]} := \text{reshape}^{-1}(\hat{U}) \in \mathbb{C}^{b_i \times d_i \times b_{i + 1}},\quad
    M_{[i+2]} := \text{reshape}^{-1}(\hat{R}) \in \mathbb{C}^{b_{i+1} \times d_{i + 1} \times b_{i + 2}}\label{eq:order_4_svd_dimensions}
\end{align}
Because $\hat{U}^\dagger \hat{U} = I_{s_i}$, the tensor $A_{[i+1]}$ is left-normalized, and the orthogonality center has moved to site $i+2$. By contraction over the joint index $b_{i+1}$ we can verify $T= A_{[i+1]}\cdot M_{[i+2]}$, if the SVD is performed without truncation. 

In the case of $\epsilon$-truncation in the SVD, the bond dimension is reduced to $b_{i + 1} < s_i$, where $b_{i+1}\geq 1$ is set as the largest index that satisfies the inequality \eqref{eq:eq_epsilon} for a given $\epsilon$.
The discarded singular values lead to the approximation $\widetilde{T} := A_{[i+1]}\cdot M_{[i+2]} \approx {T}$, where
\begin{align}
    \|T - \widetilde{T} \|_F = \| \hat{U}\hat{\Sigma}\hat{V}^{\dag} - \hat{U}\widetilde{\Sigma}\hat{V}^{\dag} \|_F = \|\hat{\Sigma} - \widetilde{\Sigma} \|_F \leq \epsilon.
\end{align}
Here, $\widetilde{\Sigma}$ is the approximate singular value matrix where the $b_{i+1}$ largest singular values of $\hat{\Sigma}$ are retained.

\bibliography{references}

\end{document}